\documentclass[journal]{IEEEtran}

\usepackage{ucs}
\usepackage[cmex10]{amsmath,mathtools}
\usepackage{cite, amsfonts, amssymb, amsthm, bm, bbm, graphicx, relsize, multirow, booktabs, tikz,subfigure,soul}
\usepackage[american]{babel}
\usepackage[T1]{fontenc}
\usepackage{algorithmic, algorithm}
\usepackage[multiple]{footmisc}
\newtheorem{theorem}{Theorem}

\theoremstyle{definition}

\theoremstyle{remark}

\newcommand{\ds}{\displaystyle}
\newcommand{\norm}[1]{\left\lVert#1\right\rVert}

\begin{document}
This paper has been accepted for publication on the IEEE Transactions on Wireless Communications

\copyright 2020 IEEE. Personal use of this material is permitted. Permission from IEEE must be obtained for all other uses, in any current or future media, including reprinting/republishing this material for advertising or promotional purposes, creating new collective works, for resale or redistribution to servers or lists, or reuse of any copyrighted  component of this work in other works.”

\newpage

\title{Communications and Radar Coexistence in the Massive MIMO Regime: Uplink Analysis}

\author{Carmen D'Andrea, {\em Student Member}, {\em IEEE}, Stefano Buzzi, {\em Senior Member}, {\em IEEE}, and  Marco Lops,
{\em Fellow}, {\em IEEE}
\thanks{This paper was partly presented at the \textit{19th IEEE International Workshop on
Signal Processing Advances in Wireless Communications}, Kalamata, Greece, June 2018, and will be partly presented at the \textit{44th IEEE International Conference on Acoustics, Speech, and Signal Processing}, Brighton, U.K., May 2019.}
\thanks{C. D'Andrea and S. Buzzi are with the Department of Electrical and Information Engineering, University of Cassino and Lazio Meridionale, I-03043 Cassino, Italy (\{carmen.dandrea, buzzi\}@unicas.it). M. Lops is with the Department of Electrical Engineering and Information Technologies, University "Federico II" of Naples, Naples, Italy (lops@unina.it). The work of C. D'Andrea and S. Buzzi has been supported by the MIUR program ``Dipartimenti di Eccellenza 2018-2022". }
}
\maketitle

\begin{abstract}
This paper considers the uplink of a massive MIMO communication system using 5G New Radio-compliant multiple access, which has to co-exist with a radar system using the same frequency band.  A  system model taking into account the reverberation (clutter) produced by the radar system onto the massive MIMO receiver is proposed. In this scenario, several receivers for uplink channel estimation and data detection are proposed, ranging from the simple channel-matched beamformer to the zero-forcing and linear minimum mean square error receivers for clutter disturbance rejection, under the two opposite situations of perfectly known and completely unknown clutter covariance.
A theoretical analysis is also provided, deriving a lower bound on the achievable uplink spectral efficiency and the mutual information between the input Gaussian-encoded symbols and the observables available at the communication receiver of the cellular massive MIMO system: regarding the latter, in particular, it is shown that, 
in the large antenna number regime, and under the assumption of perfect channel state information (CSI), the effect of radar clutter at the base station is suppressed and single-user capacity may be restored.  
Numerical results, illustrating the performance of the proposed detection schemes, confirm the findings of the theoretical analysis, and permit quantifying the system robustness to clutter effect for increasing number of antennas at the base station.
\end{abstract}


\section{Introduction}
Radar-Communications co-existence in the same frequency band 
has recently aroused a vibrant academic and industrial interest \cite{zhengSPM}, since it represents one of the key enabling technologies to second the inevitable scaling
up of the carrier frequencies of terrestrial networks \cite{griffiths2015radar}. In fact, the standard evolution from GSM to the fifth generation (5G) has produced a progressive invasion of frequency bands traditionally used by radar systems. Such key words as {\em spectrum sharing}, {\em Dual Function Radar Communication (DFRC)}, {\em Convergence} have become commonplace in the technical jargon to denote the different philosophies and architectures introduced so far in this area. It is in fact anticipated that in the near future 
not only the $2-8$GHz frequency range, comprising the traditional $S$ and $C$ radar bands, but also the 24 GHz and 60 GHz bands, devoted to very high resolution mapping, scientific remote sensing and airport (short-range) surveillance, will be inevitably used for both communication and sensing.  Accordingly, the Defense Advanced Research Projects Agency (DARPA)  recently announced the Shared SPectrum Access for Radar and Communications
(SSPARC) program \cite{evans2016shared}. 
A possible classification of the approaches proposed so far might follow the taxonomy proposed in \cite{zhengSPM}, wherein the major categorization is between 
architectures where both the radar and the communication system have active transmitters and those wherein transmission takes place in a unique integrated platform, thus allowing a joint co-design.

To the former family belong both {\em selfish} and {\em holistic} architectures. Selfish design, in particular, focuses the attention on one system, i.e., communication or radar, and adopts strategies to counteract the interference induced by the spectral overlap without paying attention to the performance of the other co-existing system. Selfish architectures include radar-centric systems \cite{Deng2013Interference,Aubry2014Radar,Aubry2015Optimizing}, or communication-centric structures, where the interference induced by the radar is dealt with either at the receiver \cite{JSTSP} or, in the presence of some CSI, directly at the transmitter \cite{tuninetti}. Holistic architectures, conversely, rely on the concept of heavy cooperation between the transmitting systems, whereby the communication codebook and the radar waveform(s) are jointly designed, so as to guarantee the performance of both systems. This idea, first proposed in \cite{li2016joint}, has been successively developed to account for a number of possible scenarios, and in particular for the reverberation produced by the radar (clutter) on its own receiver and on the communication receiver \cite{lipetropulu,li2016mimo,Li2016Optimum,qian2018joint}. Some form of cooperation is also assumed in a class of systems which borrow channel sensing techniques from the cognitive radio literature to detect and exploit spectral holes so as to allow co-existence with no spectral overlap \cite{CME18}.

On the other hand, a completely different philosophy is to grant {\em functional co-existence} without generating mutual interference; the transmission phase for all co-existing systems takes place in a unique integrated platform, thus leading to the aforementioned concept of DFRC \cite{blunt2010embedding,Blunt,hassanien2016dual,liu2018dual}. In this case,  the communication signal is typically {\em embedded} in the radar signal, by exploiting either its latency or the transmit antenna beam side-lobes, so that no real spectral overlap takes place.

Unfortunately, none of the above approaches appears applicable if - as it will be the case with  5G-and-beyond systems - the wireless network is to be {\em added} to pre-existing sensing systems and full cooperation cannot be realized, due to, e.g., security reasons. On the other hand, a consensus has now been reached on the fact that one of the most damaging effects of co-existence is the clutter produced by a search radar onto the base-station of the wireless network, which ultimately may result in a dramatic reduction of the uplink rates. Under these circumstances, also in consideration of the different order of magnitude of the powers in play, there is no prior guarantee of the feasibility of a full spectral overlay. 

\subsection{Paper contribution}
The aim of the present contribution is to demonstrate that a 5G wireless network, employing a standard Orthogonal Frequency Division Multiplexing (OFDM) modulation format and endowed with a {\em massive} MIMO array at the base station may successfully co-exist with a wide-beam search radar, taking huge advantage of the massive nature of the receive array. Massive MIMO was introduced by Marzetta, in his pioneering paper \cite{Marzetta10}; this technology represents a solid milestone of current and future wireless systems \cite{MassiveMIMO_book_Bjornson,marzetta2016fundamentals}. Massive MIMO amounts to use a very large number of service antennas (e.g., hundreds or thousands) in order to serve a lower number of mobile users
with the time-division-duplex (TDD) protocol so as to exploit uplink/downlink channel reciprocity. In particular, our focus is on the effect that the massive structure may - or may not - have on clutter mitigation in the two relevant phases of the uplink haul, i.e. the {\em training phase} for user channel acquisition and the {\em demodulation phase} for data transmission. 
To the best of authors' knowledge, this is the first paper to study the robustness of a massive MIMO cellular system to the radar interference co-existing in the same frequency bands, while preliminary investigations on this issue have appeared in \cite{radar_massiveMIMO1} and
\cite{radar_massiveMIMO2}.
The contribution of this paper can be thus summarized as follows. 
First of all, 
inspired by the 5G standard, we consider a Single-Carrier (SC) FDMA operating at a carrier frequency of $3$GHz, and a co-existing radar system employing a {\em sophisticated}\footnote{This term indicates that the duration-bandwidth product of the radar waveform is considerably larger than one.} waveform with the same bandwidth. We present a model for the signal received at the Base Station (BS) array, accounting for the effect of the radar reflections on the whole set of packets entering the radar Pulse Repetition Time (PRT).
We then perform an information-theoretic analysis showing that the massive MIMO structure, under the assumption of perfect CSI and single-user transmission, 
is intrinsically resistant to the clutter effect in the limit of arbitrarily large number of antennas at the BS; in particular we show that the mutual information between the Gaussian-distributed information symbols and the observables at the BS antenna array becomes independent of the clutter contribution in the limit of large number of antennas. 
Next, several practical receivers for uplink channel estimation and data detection at the BS are proposed, examining both the case in which the clutter second-order statistics are known to the BS receiver, and the case in which no prior knowledge about the clutter can be assumed. A lower bound to the system Spectral Efficiency (SE) is also analytically derived. 
Extensive simulation results are finally provided  in order to corroborate the analytic findings and to show the performance of the different proposed channel estimation and data detection structures. 

This paper is organized as follows. In the next section we illustrate the considered system model,
along with the model of the received signal at the BS, both in the case of uplink data transmission and uplink pilot transmission for channel estimation. Section III contains the information-theoretic analysis of the system, while 
Section IV is devoted to the derivation of the considered  uplink channel estimation and data detection structures. In Section V the derivation of the lower bound to the system SE is reported, while numerical results are shown and discussed in Section VI. Finally, concluding remarks are given in 
Section VII.

\subsection{Notation}
\noindent
The following notation is used in the paper. The transpose, the inverse and the conjugate transpose of a matrix $\mathbf{A}$ are denoted by $\mathbf{A}^T$, $\mathbf{A}^{-1}$ and $\mathbf{A}^H$, respectively. The generalized Moore-Penrose inverse of a matrix $\mathbf{A}$ is denoted by $\mathbf{A}^{\dagger}$. The trace  and the determinant of the matrix $\mathbf{A}$ are denoted as tr$\left(\mathbf{A}\right)$ and det$\left(\mathbf{A}\right)$, respectively. The $N$-dimensional identity matrix is denoted as $\mathbf{I}_N$, the $(N \times M)$-dimensional matrix with all zero entries is denoted as $\mathbf{0}_{N \times M}$ and $\mathbf{1}_{N \times M} $ denotes a $(N \times M)$-dimensional matrix with unit entries. The vectorization operator is denoted by vec$(\cdot)$ and the Kronecker product is denoted by $\otimes$. The $(m,\ell$)-th entry  and the $\ell$-th column of the matrix $\mathbf{A}$ are denoted as $\left[\mathbf{A}\right]_{(m,\ell)}$ and $\left[\mathbf{A}\right]_{(:,\ell)}$, respectively. The block-diagonal matrix obtained from matrices $\mathbf{A}_1, \ldots, \mathbf{A}_N$ is denoted by blkdiag$\left( \mathbf{A}_1, \ldots, \mathbf{A}_N\right)$. The Dirac's delta pulse is denoted as $\delta(t)$. The statistical expectation operator is denoted as $\mathbb{E}[\cdot]$; $\mathcal{CN}\left(\mu,\sigma^2\right)$ denotes a complex circularly symmetric Gaussian random variable with mean $\mu$ and variance $\sigma^2$.

\section{System model} \label{system_model_section}
\newcommand{\bp}{\mathbf{p}}
\newcommand{\bx}{\mathbf{x}}
\newcommand{\bX}{\mathbf{X}}
\newcommand{\bP}{\mathbf{P}}
\newcommand{\bW}{\mathbf{W}}
\newcommand{\bw}{\mathbf{w}}
\newcommand{\bg}{\mathbf{g}}
\newcommand{\bh}{\mathbf{h}}
\newcommand{\bI}{\mathbf{I}}
\newcommand{\bT}{\mathbf{T}}
\newcommand{\by}{\mathbf{y}}
\newcommand{\bY}{\mathbf{Y}}
\newcommand{\bzero}{\mathbf{0}}

\begin{figure}
\begin{center}
\includegraphics[scale=1]{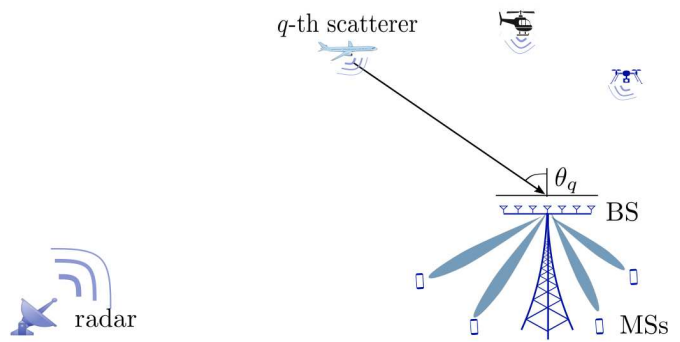}
\caption{A massive MIMO cellular system co-existing with a radar system. The BS received signal is corrupted by the clutter echoes. Ambient scatterers are seen as point-like targets placed at some random angles.}
\label{fig:scenario}
\end{center}
\end{figure}
Consider a single-cell massive MIMO communication system using SC-FDMA multiple access in the uplink, operating at a carrier frequency $f_c=3$ GHz and coexisting with a radar system using the same frequency band, as depicted in Fig. \ref{fig:scenario}.
With regard to the massive MIMO system,
we use the following notation and assumptions:
\begin{itemize}
\item[-] $N$ denotes the number of subcarriers of the SC-FDMA system ($N=4096$ will be assumed);
\item[-] The BS is equipped with a uniform linear array (ULA) with $M$ elements; fully digital beamforming is assumed, so that the number of RF chains coincides with the number of antennas.
\item[-] The mobile stations (MSs) transceivers are equipped with a single antenna, and the number of MSs in the system is $K$.
\item[-] The subcarrier spacing is denoted by $\Delta f$ ($\Delta f=30$ kHz is assumed).
\item[-] A block fading channel is assumed with channel coherence bandwidth equal to $C \Delta f$, with $C=16$. Otherwise stated, the channel can be considered constant over $C$ consecutive carriers and then takes a new value statistically independent of  the previous one. Note that, for each user, and for each BS receive antenna, CSI amounts to $Q=N/C=256$ complex scalar coefficients.
\item[-] The uplink channel between the $k$-th single-antenna MS and the BS on the $n$-th carrier is represented by the $M$-dimensional vector $\bh_k^{(\lceil n/C \rceil)}= \beta_k \bg_k^{(\lceil n/C \rceil)}$, where $\beta_k$ takes into account the path-loss and the log-normal shadowing (fully correlated across antennas and subcarriers), while $\bg_k$ denotes the small-scale fading and is a random vector with ${\cal CN}(0, \bI_M)$ distribution.
\item[-] The MSs transmit simultaneously using all the available subcarriers; user separation is performed in the spatial domain thanks to the use of a large number of antennas.
\item[-]
The uplink frame structure is depicted in Fig. \ref{fig:frame_structure}. Each packet is made of a cyclic-prefix (CP) and of a sequence of data symbols; the CP discrete length is $N_{\rm CP}=288$, while the length of the data symbols is $N$. The timing is such that $N_{\rm pkt}=14$ packets fit into a 0.5 ms timeslot, which leads to a symbol time $T_s= 8.146$ ns. These numbers are inspired by the December 2017 3GPP first release of the 5G New Radio standard.
\end{itemize}
With regard to the radar system,
the following assumptions are made.
\begin{itemize}
\item[-]
The radar operates at the same carrier frequency as the wireless cellular system and it is assumed that there is full overlap between the bandwidths of the radar signal and of the communication signals transmitted by the MSs\footnote{This assumption is made to simplify the notation; the generalization of the results of this paper to the case of partial spectral overlap can be treated with standard techniques.}.
\item[-]
The radar transmits a coded waveform, of duration $LT_s$; its baseband equivalent is expressed as
\begin{equation}
s_R(t)=\ds \sqrt{P_T} \sum_{\ell=0}^{L-1} c_\ell \psi(t-\ell T_s) \; ,
\label{eq:radarsignal}
\end{equation}
wherein $P_T$ is the radar transmitted power, $[c_0, c_1, \ldots, c_{L-1}]$ is the unit-energy radar code, and $\psi(\cdot)$ is the base pulse; we assume that $\psi(\cdot)$ is a unit-energy rectangular pulse of duration $T_s$. The value $L=32$ is assumed in this paper.
\item[-] The waveform $s_R(t)$ is transmitted periodically every $T_{\rm PRT}=1$ ms, with $T_{\rm PRT}$ the PRT; this corresponds to assuming a maximum range of 150 Km, which is customary in surveillance systems, but our derivations carry over to the case of shorter-range systems.
\end{itemize}

\medskip

In the following, we provide a model for the uplink signal received at the BS, taking into account
both the data signals transmitted by the MSs and the contribution from the radar system due to the presence of scatterers in the surrounding environment. A block scheme of the generic MS transmitter is reported in the upper part of Fig. \ref{fig:transceivers}, while the lower part of the same figure represents a block scheme of the uplink receiver at the generic receive BS antenna.
As it is seen from the frame structure in Fig. \ref{fig:frame_structure}, 14 data packets fit into a 0.5 ms time window; some of these packets can be used to transmit known training symbols in order to enable channel estimation. In the following, we describe separately the signal model for the data packets in the training phase and in the data communication phase. In order to help the reader to keep up with the paper notation, we report in Table \ref{Symbols_table} the meaning of the
main mathematical symbols used in the following.

\begin{table}[]
\caption{Meaning of the main mathematical symbols}
\label{Symbols_table}
\begin{tabular}{|p{1.5cm}|p{6.5cm}|}
\hline
\textbf{Symbol} & \textbf{Interpretation} \\ \hline
$\mathbf{X}_k(\ell)^{(n)}$ & $n$-th coefficient of the isometric FFT of the symbol from the $k$-th user in the $\ell$-th packet \\ \hline
$\mathbf{h}_k^{(\lceil n/C \rceil)}$ &  uplink channel between the $k$-th single-antenna MS and the BS on the $n$-th carrier \\ \hline
$\mathbf{W}(\ell)^{(n)}$ & additive thermal noise on the $n$-th subcarrier in the $\ell$-th data packet \\ \hline
$ \mathbf{C}(\ell)^{(n)}$ & clutter contribution on the $n$-th subcarrier in the $\ell$-th data packet \\ \hline
$\mathcal{Y}_q,\mathcal{C}_q, \mathcal{W}_q$ & vectors containing data, clutter and noise, respectively, for the estimation of the $q$-th channel realization, with $q=\lceil \frac{n}{C}\rceil$ \\ \hline
$\mathbf{P}_k^{(q)}$ & vector containing the FFT coefficients of the pilot sequence for the $k$-th user devoted to the estimation of the $q$-th channel realization \\ \hline
$\mathbf{b}\left( \theta \right)$ & BS ULA array response for the generic angle $\theta$ \\ \hline
$\widehat{\mathbf{h}}_k^{(q)}$ &  estimate of the $q$-th realization of the uplink channel between the $k$-th single-antenna MS and the BS \\ \hline
$\mathbf{D}_{q,k}$ &  MMSE estimation matrix for the $q$-th realization of the uplink channel between the $k$-th single-antenna MS and the BS \\ \hline
$\mathbf{F}$ & permutation matrix such that $\text{vec} \left(\mathcal{C}_q^T \right)=\mathbf{F} \text{vec} \left(\mathcal{C}_q \right)$ \\ \hline
$\mathbf{K}_{\mathbf{C}(\ell)^{(n)}}$ & covariance matrix of the clutter vector $\mathbf{C}(\ell)^{(n)}$ \\ \hline
$\mathbf{v}_k^{(q)}$ &  generic combining vector for the symbol transmitted on the $n$-th subcarrier from the $k$-th user, with $q=\lceil \frac{n}{C}\rceil$  \\ \hline
\end{tabular}
\end{table}

\begin{figure}
\begin{center}
\includegraphics[scale=0.2]{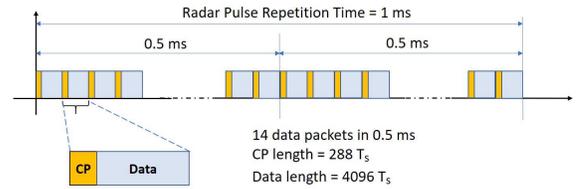}
\caption{Uplink frame structure. Data packets are made of a CP (of length 288 in discrete samples) and of information symbols (of length $N=4096$ in discrete samples). The symbol time is such that $N_{\rm pkt}=14$ data packets fit into 0.5 ms. The radar PRT is 1 ms.}
\label{fig:frame_structure}
\end{center}
\end{figure}

\begin{figure}
\begin{center}
\includegraphics[scale=0.25]{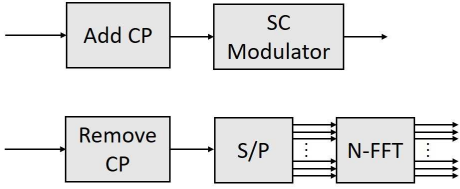}
\caption{Upper figure: Block-scheme of the transmitter at the generic mobile station. Lower figure: Block-scheme of the BS receiver at the generic antenna; assuming fully-digital beamforming at the BS, this scheme is to be replicated for each receive antenna.}
\label{fig:transceivers}
\end{center}
\end{figure}

\subsection{Signal model during uplink data transmission}
Consider the generic $\ell$-th data packet; denote by $\mathbf{x}_k(\ell)$ the $N$-dimensional vector containing data symbols from the $k$-th MS to be transmitted in the $\ell$-th data packet; denote by $\mathbf{X}_k(\ell)$ the $N$-dimensional vector representing the isometric FFT of $\mathbf{x}_k(\ell)$.
In the following, for the sake of simplicity, we will focus on the problem of detecting the FFT-ed symbols 
$\mathbf{X}_k(\ell)$ in place of the original symbols $\mathbf{x}_k(\ell)$. Given the isometric, orthogonality-preserving, relationship between $\mathbf{x}_k(\ell)$ and $\mathbf{X}_k(\ell)$, this assumption does not imply any loss of generality.
 Referring to the lower part of Fig. \ref{fig:transceivers}, it is easily shown that the observable corresponding to the $n$-th subcarrier after the FFT operation can be represented through the following $M$-dimensional vector:
\begin{equation}
\mathbf{y}(\ell)^{(n)}= \ds \sum_{k=1}^K
\sqrt{p_k} \mathbf{X}_k(\ell)^{(n)} \mathbf{h}_k^{(\lceil n/C \rceil)}+ \mathbf{W}(\ell)^{(n)} + \mathbf{C}(\ell)^{(n)} \; ,
\label{eq:vector_observable}
\end{equation}
for $n=1, \ldots, N$. In the above equation, $p_k$ is the power transmitted by the $k$-th MS, $\mathbf{X}_k(\ell)^{(n)}$ is the $n$-th entry of the vector $\mathbf{X}_k(\ell)$,  $\mathbf{W}(\ell)^{(n)}$ is a $\mathcal{CN}(\mathbf{0},
\sigma_w^2 \mathbf{I}_M)$ random vector representing the additive thermal noise, while
$\mathbf{C}(\ell)^{(n)}$ is the clutter contribution generated by the radar system on the $n$-th subcarrier; an expression for such vector will be given in the following.
Grouping together the data corresponding to the $N$ subcarriers we finally get the following $(M \times N)$-dimensional matrix for the observables corresponding to the $\ell$-th data packet:
\begin{equation}
\begin{array}{lll}
\mathbf{Y}(\ell)=\!\!\!\!& \ds \sum_{k=1}^K \sqrt{p_k}
\left( \left[\mathbf{h}_k^{(1)} \ldots \mathbf{h}_k^{(Q)}\right] \otimes \mathbf{1}_{1 \times C} \right) \mbox{diag}( \mathbf{X}_k(\ell))  \\ & +\mathbf{W}(\ell) + \mathbf{C}(\ell) \; .
\end{array}\label{eq:observable_data}
\end{equation}

\subsection{Signal model during uplink training}
Consider now the case in which the MSs transmit known pilot sequences to enable channel estimation at the BS. Let $T$ denote the number of consecutive packets devoted to training, and let $\mathbf{p}_k(1), \ldots, \mathbf{p}_k(T)$ denote $N$-dimensional vectors containing the $k$-th MS pilots to be used in the $T$ packets used for channel estimation.
Focusing on the $\ell$-th packet (with now $\ell=1, \ldots, T$), and following the same steps as in the previous section, it is easily shown that the observable at the output of the FFT block at the BS receiver can be written as the following $(M \times N)$-dimensional matrix
\begin{equation}
\begin{array}{llll}
\mathbf{Y}(\ell)= & \ds \sum_{k=1}^K \sqrt{p_{\rm{p},k}}
\left( \left[\mathbf{h}_k^{(1)} \ldots \mathbf{h}_k^{(Q)}\right] \otimes \mathbf{1}_{1 \times C} \right)  \\ & \mbox{diag}( \mathbf{W}_{N, FFT}\mathbf{p}_k(\ell))+ \mathbf{W}(\ell) + \mathbf{C}(\ell) \; ,\end{array}
\label{eq:observable_estimation}
\end{equation}
where, now, $p_{\rm{p},k}$ is the power transmitted by the $k-$th user during the uplink training phase, and $\mathbf{W}_{N, FFT}$ is the $(N \times N)$-dimensional matrix performing an isometric FFT\footnote{The $(m,n)$-th entry of $\mathbf{W}_{N, FFT}$ is thus $\frac{1}{\sqrt{N}} e^{-j 2 \pi (m-1)(n-1)/N}$.}.
Assume now that the $M$-dimensional channel vectors $\mathbf{h}_k^{(q)}$, $\forall k=0, \ldots, K-1$, have to be estimated; to this end, the columns from the $[(q-1)C+1]$-th to the $[qC]$-th of the matrices $\mathbf{Y}(1), \ldots, \mathbf{Y}(T)$ must be processed; they form the following observable:
\begin{equation}
\mathcal{Y}_q =  \ds \sum_{k=1}^K\sqrt{p_{\rm{p},k}}\mathbf{h}_k^{(q)} \mathbf{P}_k^{(q) \, T}
+ \mathcal{W}_q + \mathcal{C}_q \; ,
\label{eq:observable_estimation_q}
\end{equation}
where
$$
\mathcal{W}_q=\left[ \mathbf{W}(1)_{:,(q-1)C+1:qC} \cdots \mathbf{W}(T)_{:,(q-1)C+1:qC}\right]  \; ,
$$
$$\mathcal{C}_q=\left[ \mathbf{C}(1)_{:,(q-1)C+1:qC} \cdots \mathbf{C}(T)_{:,(q-1)C+1:qC}\right] \; , $$
and $\mathbf{P}_k^{(q)}$ is a $(TC)$-dimensional vector defined as follows:
\begin{equation}
\begin{array}{lll}
\mathbf{P}_k^{(q)} \triangleq &
\left[ \left(\mathbf{W}_{N, FFT}\mathbf{p}_k(1)\right)_{(q-1)C+1:qC}, \right. \\ &
\left. \ldots ,  \left(\mathbf{W}_{N, FFT}\mathbf{p}_k(T)\right)_{(q-1)C+1:qC}\right]^T \; . \label{eq:vector_pilot}
\end{array}
\end{equation}

\subsection{Clutter modeling}
\newcommand{\bR}{\mathbf{R}}\newcommand{\br}{\mathbf{r}}
We now illustrate the clutter model and provide an explicit expression for the $(M \times N)$-dimensional clutter matrix $\mathbf{C}(\ell)$ affecting the $\ell$-th received data packet.

The clutter disturbance is actually generated by a large set of discrete scatterers in the surrounding environment. Given the BS array dimension, these scatterers are seen by the BS as "co-located" \cite{JianLi}, namely all of the antennas see the scatterer under the same aspect angle and with the same (complex) scattering coefficient\footnote{In fact, two antennas spaced a distance $d$ apart and tuned to a wavelength $\lambda$ see a target/scatterer located at distance $R$ and having an extension $V$ in the antenna alignment direction under the same aspect angle iff $d<\frac{\lambda R}{V}$, which is for sure true in the scenario considered here, where the antenna spacing is in the order of centimeters, the scatterers are point-like, and their distance from the BS may be in the order of kilometers.}. The radar-to-BS channel can be henceforth modeled as an LTI system with the following vector-valued impulse response:
\begin{equation}
\mathbf{h}(t)= \ds \sum_{q=0}^{N_s-1} \sum_{m=0}^{Q-1} \beta_{q,m} \mathbf{b}(\theta_q) \delta(t - \tau_q - m/W) \; .
\label{eq:radar-to-BS}
\end{equation}
In the above equation, $N_s$ denotes the number of scatterers in the surrounding environment that contribute to the clutter disturbance; $\theta_q$ is the direction of arrival of the clutter contribution from the $q$-th scatterer, $\tau_q$ is the propagation delay associated to the signal generated by the $q$-th scatterer and the BS ULA array response for the generic angle $\theta$ is defined as
\begin{equation}
\mathbf{b}(\theta)=[1 \; e^{-j 2 \pi d \sin(\theta)/\lambda} \; \ldots e^{-j (M-1) 2 \pi d \sin(\theta)/\lambda }]^T \, . 
\label{BS_ULA}
\end{equation}
Moreover, since the signal bandwidth $W$ exceeds the channel coherence time, we also assume that each physical scatterer generates $Q$ clutter echoes spaced integer multiples of $1/W$ apart; accordingly, $\beta_{q,m}$ is the reflection coefficient associated to the $m$-th replica from the $q$-th scatterer.

Now, recall that the radar transmits the waveform in \eqref{eq:radarsignal}; this waveform travels through a channel with the impulse response  $\mathbf{h}(t)$  in \eqref{eq:radar-to-BS}  and then is passed through a filter with a rectangular impulse response of duration $T_s$ and sampled at rate $1/T_s$. After A/D conversion, the baseband equivalent of the clutter disturbance can be represented as the following vector-valued discrete-time (sampled at rate $1/T_s$) signal:
\begin{eqnarray}
\widetilde{\mathbf{s}}_R(\eta)=\ds
\sum_{q=0}^{N_s-1} \sum_{m=0}^{Q-1} \sum_{p=0}^{L-1}
\sqrt{P_T}\beta_{q,m}c_p \mathbf{b}(\theta_q) \nonumber \\
r_{\psi} ((\eta-p)T_s-m/W -\tau_q) \; ,
\label{eq:discrete-time-clutter}
\end{eqnarray}
with $r_{\psi}(\cdot)$ the autocorrelation function of the base pulse.

Now, refer to the frame structure of Fig. \ref{fig:frame_structure} and assume, for simplicity, that the radar transmits its signal at the beginning of a 0.5 ms timeframe\footnote{This assumption can be removed with standard techniques.}.  Denoting by $T_{\rm pkt}=(4096+288)T_s$ the duration of a data packet including its CP, the generic $\ell$-th packet starts at time $\ell T_{\rm pkt} + T_{\rm CP}$ and ends at $(\ell+1) T_{\rm pkt}$.
Let now ${\cal S}(\ell)$ denote the set of the scatterers corrupting the reception of the $\ell$-th data packet. Since the contribution from the generic $q$-th scatterer starts at $\tau_q$ and stops at $\tau_q+QT_s+LT_s$, it is easily seen that the set ${\cal S}(\ell)$ can be defined as
\begin{equation}
\begin{array}{lll}
{\cal S}(\ell)= \left\{ q \in \{0, 1, \ldots, N_s-1\} \; : \; \right. \\  \left.
[\tau_q, \tau_q+QT_s+LT_s] \, \cap \,[\ell T_{\rm pkt} + T_{\rm CP}, (\ell+1) T_{\rm pkt}] \neq \emptyset \right\}\; ,
\end{array}
\end{equation}
with $\emptyset$ denoting the empty set. Using the above notation, the clutter $(M \times N)$-dimensional matrix appearing in Eqs. \eqref{eq:observable_data} and \eqref{eq:observable_estimation} can be expressed as
\begin{equation}
\mathbf{C}(\ell)=\ds \sum_{q \in {\cal S}(\ell)}
\sum_{m=0}^{Q-1} \sum_{p=0}^{L-1}\sqrt{P_T}\beta_{q,m} c_p \mathbf{b}(\theta_q) \mathbf{r}^T_{q,p,m}(\ell) \mathbf{W}_{N, FFT} \; ,
\label{clutter_matrix}
\end{equation}
wherein
\begin{equation}
\begin{array}{lll}
\mathbf{r}_{q,p,m}(\ell)= & \left[ r_{\psi}\left({\ell  T_{\rm pkt} + T_{\rm CP }} +T_s-pT_s -\frac{m}{W} - \tau_q \right) \, ,  \right. \\ & \left.  \ldots, \,
r_{\psi}\left({(\ell +1) T_{\rm pkt} } -pT_s -\frac{m}{W} - \tau_q \right)
\right]^T \; .
\end{array}
\end{equation}
For future reference, we define the $N-$dimensional row vector
\begin{equation}
{\widetilde{\bR}}_{q,\ell,m}^T=\sum_{p=0}^{L-1} c_p\mathbf{r}^T_{q,p,m}(\ell) \mathbf{W}_{N, FFT}\; .
\label{R_tilde_definition}
\end{equation}

\section{Information-theoretic analysis}
In this section we show that, assuming perfect CSI, 
the radar clutter contribution to the mutual information between the observable and the user information symbols vanishes in the mean-square sense in the massive MIMO regime, i.e. for $M \rightarrow +\infty$.
Even though the result could be shown in the general case, we particularize our analysis to the simple case of single-user transmission and focus the attention on a single packet.
The latter assumption is legitimate since the presence of a cyclic prefix guarantees that no scattering center can affect consecutive OFDM symbols. 
In order to simplify the notation, we omit the user and the packet indexes. Under these circumstances, let us consider $C$ consecutive sub-carriers, extending from $(n-1)C+1$ to $nC$, which experience the same channel fading. Thus, \eqref{eq:observable_data} simplifies to:
\begin{equation}
\begin{array}{llll}
\bm y_{(n-1)C+1:nC}=& \sqrt{p} \bm X^{(n-1)}\otimes \bm h^{(n-1)}+\bm{\widetilde{w}}_{(n-1)C+1:nC} \\ &\ds + \bm{\widetilde{c}}_{(n-1)C+1:nC} \in \mathbb{C}^{CM}\, , \\ \ds \bm X^{(n-1)} \in \mathbb{C}^C\, , & \bm h^{(n-1)} \in \mathbb{C}^M
\end{array}
\end{equation}
where $\bm X^{(n-1)}=\left[X^{(n-1)C}, \ldots, X^{nC-1}\right]^T$, $\bm y=\text{vec} \left( \bm Y\right)$, $\bm{\widetilde{w}}=\text{vec}\left( \bm W\right)$, $\bm{\widetilde{c}}=\text{vec}\left( \bm C\right)$,  
and $\bm a_{i:j}$ denotes the entries of the vector $\bm a$ from the $i$-th to the $j$-th one.
We can thus form the $QCM=NM$-dimensional vector
\begin{equation}
\bm y=\left[\bm y^T_{1:C}, \bm y^T_{C+1:2C}, \ldots , \bm y^T_{(Q-1)C+1:QC}\right]^T
\end{equation}
In order to study the clutter effect on the massive MIMO system, we start evaluating the clutter covariance matrix of the whole clutter-related observable $\bm{\widetilde{c}}$. To this end,  we assume to have $\widetilde{N}_s$ scatterers, each contributing $Q$ replicas of the radar signal in the packet under test. Under these circumstances, and using the definitions in \eqref{clutter_matrix} and \eqref{R_tilde_definition}, the clutter matrix in the packet under test can be written as
\begin{equation}
\mathbf{C}=\ds \sum_{q =1}^{\widetilde{N}_s}
\sum_{m=0}^{Q-1} \sqrt{P_T}\beta_{q,m}  \mathbf{b}(\theta_q) \mathbf{\widetilde{R}}^T_{q,m} \; ,
\label{clutter_matrix_generic}
\end{equation}
which leads to the following expression for the 
 $MN$-dimensional vector $\bm{\widetilde{c}}$:
\begin{equation}
\bm{\widetilde{c}}=\ds \sum_{q =1}^{\widetilde{N}_s}
\sum_{m=0}^{Q-1} \sqrt{P_T}\beta_{q,m}  \mathbf{\widetilde{R}}_{q,m} \otimes \mathbf{b}(\theta_q) \; .
\label{clutter_vector_generic}
\end{equation}
Accordingly, the $(MN \times MN)$-dimensional clutter covariance matrix from $\widetilde{N}_s$ scatterers whose replicas are {\em all} contained in a given packet can be written as
\begin{equation}
\bm K_{\rm c}= \mathbb{E}\left[ \bm{\widetilde{c}} \bm{\widetilde{c}}^H\right]=\sum_{q=1}^{\widetilde{N}_s}\sum_{m=0}^{Q-1} P_T \sigma^2_{q,m} \mathbf{\widetilde{A}}_{q,m} \otimes \bm b(\theta_q) \bm b^H(\theta_q) \, , 
\label{eq:clutter_covariance}
\end{equation}
where the above result is obtained exploiting the properties of the Kronecker product and defining the $(N \times N)-$dimensional matrix $\mathbf{\widetilde{A}}_{q,m}= \mathbf{\widetilde{R}}_{q,m} \mathbf{\widetilde{R}}_{q,m}^H$.
Consider now the uplink data transmission phase; the single-user mutual information has the general form\cite{Cover-Thomas}
\begin{equation}
\begin{array}{llll}
&I \left(\left.\bm y;\bm X^{(0)}, \bm X^{(1)}, \ldots , \bm X^{(Q-1)}\right|\bm h^{(0)}, \ldots, \bm h^{(Q-1)} \right)= \\ &\log \det \left(\sigma^2_w \bm I_{NM}+\bm K_{\rm c}+\bm K'\right)-\log \det
\left(\sigma^2_w \bm I_{NM}+\bm K_{\rm c}\right)=\\
&\log \det \left[\bm I_{NM}+\left(\sigma^2_w \bm I_{NM}+\bm K_{\rm c} \right)^{-1}\bm K' \right] \, ,
\end{array}
\label{SU_mutual_inf}
\end{equation}
where $\bm K'$ is expressed in \eqref{eq:K'} on top of the next page.
\begin{figure*}
\begin{equation}
\bm K'=  p
\left(
\begin{array}{lll}
\mathbb{E} \left[\bm X^{(0)} \bm X^{(0), \, H} \right]\otimes \bm h^{(0)} \bm h^{(0), \, H} & \ldots & \mathbb{E}\left[\bm X^{(0)} \bm X^{(Q-1), \, H}\right] \otimes \bm h^{(0)} \bm h^{(Q-1), \, H} \\
\cdots & \cdots & \cdots \\
\mathbb{E} \left[\bm X^{(Q-1)} \bm X^{(0), \, H}\right] \otimes   \bm h^{(Q-1)} \bm h^{(0), \, H}  & \ldots & \mathbb{E}\left[ \bm X^{(Q-1)} \bm X^{(Q-1), \, H}\right] \otimes \bm h^{(Q-1)} \bm h^{(Q-1), \, H}
\end{array}
\right)
\label{eq:K'}
\end{equation}
\hrulefill
\end{figure*}
Given the mutual information in \eqref{SU_mutual_inf}, we can prove the following result.
\begin{theorem}
The following relation holds, with convergence in the mean square sense:
\begin{equation}
\begin{array}{llll}
\ds \lim_{M  \rightarrow \infty} \log \det \left[\bm I_{NM}+\left(\sigma^2_w \bm I_{NM}+\bm K_{\rm c} \right)^{-1}\bm K' \right] =  \\ \ds \lim_{M  \rightarrow \infty} \log \det \left[\bm I_{NM}+\frac{1}{\sigma^2_w}\bm K' \right]\;.
\end{array}
\label{eq:main_result}
\end{equation}
Otherwise stated, the clutter effect on the single user mutual information in Eq. \eqref{SU_mutual_inf} vanishes in the mean square sense in the limit $M \rightarrow \infty$.
\end{theorem}

\begin{proof}
To begin with, let us evaluate, for finite $M$,  the right-hand-side (RHS) of \eqref{eq:main_result},
which represents the single-user mutual information in \eqref{SU_mutual_inf} in the absence of radar clutter disturbance. 
Under the hypothesis that the channel vectors $\bm h^{(0)}, \ldots, \bm h^{(Q-1)}$ are linearly independent\footnote{This assumption is fulfilled with probability 1 for $M>Q$.}, 
and assuming independent information symbols\footnote{This assumption is not strictly necessary for the theorem to hold, and is made here to make the proof simpler.}, so that 
$\mathbb{E} \left[\bm X^{(i)} \bm X^{(j), \, H} \right]= \bm I_C \delta_{i,j}$, 
the matrix $\bm K'$, reported in \eqref{eq:K'}, has $CQ=N$ non-zero eigenvalues. In particular,
there are only $Q$ (possibly) distinct eigenvalues, each with multiplicity $C$. The following relation thus holds:
\begin{equation}
\log \det \left[\bm I_{NM}+\frac{1}{\sigma^2_w}\bm K' \right]=\log \det \left[\bm I_{N}+\frac{1}{\sigma^2_w} \bm \Gamma'_N (M) \right] \; , \label{eq:MI}
\end{equation}
with $\bm \Gamma'_N(M) = \text{diag}\left[\lambda_1 (\bm K'), \ldots , \lambda_N (\bm K')\right]$ a diagonal matrix 
containing the $N$ non-zero eigenvalues of $\bm K'$. 
The RHS of \eqref{eq:main_result}, for finite $M$,  can be thus expressed as
\begin{equation}
\sum_{i=1}^N\log  \left[\left(1+\frac{1}{\sigma^2_w} \lambda_i(\bm K')\right) \right]\; .
\label{eq:MI-1}
\end{equation}
For large $M$, the channel vectors $\bm h^{(0)}, \ldots, \bm h^{(Q-1)}$ become orthogonal and the eigenvalues 
of $\bm K'$ converge to 
\[
\left[\underbrace{\|{\bm h^{(0)}}\|^2, \ldots, \|{\bm h^{(0)}}\|^2}_{C \; \text{times}}, \ldots,
\underbrace{\|{\bm h^{(Q-1)}\|^2}, \ldots, {\|\bm h^{(Q-1)}\|^2}}_{C \; \text{times}}
\right] \; ,
\]
and we have
\begin{equation}
\begin{array}{lll}
\ds \lim_{M  \rightarrow \infty} \log \det \left[\bm I_{NM}+\frac{1}{\sigma^2_w}\bm K' \right]= \\ 
\ds \lim_{M  \rightarrow \infty} C\sum_{i=0}^{Q-1}\log  \left[1+\frac{p\parallel \bm h^{(i)}\parallel^2}{\sigma^2_w}\right]  \xrightarrow{\text{large}\; M} \\
CQ \ds \lim_{M  \rightarrow \infty} \ds \log  \left[1+M\frac{p \beta}{\sigma^2_w}\right]
\; ,
\end{array}
\label{eq:limit001}
\end{equation}
with $\beta$ the path-loss and shadowing coefficient. Given the fact that $\mathbb{E} \left[\bm X^{(i)} \bm X^{(j), \, H} \right]= \bm I_C \delta_{i,j}$,  the matrix $\bm K'$ in \eqref{eq:K'} can be expressed as $\bm K'= \widetilde{\bm H} \widetilde{\bm H}^H$, where $\widetilde{\bm H}$ is the following
$(NM \times N)$-dimensional matrix
\begin{equation}
\widetilde{\bm H}=\sqrt{p}
 \left[ \begin{array}{llll} \bm I_C \otimes\bm h^{(0)} & \bm 0_{MC\times C} &\ldots & \bm 0_{MC\times C} \\
\bm 0_{MC\times C} &  \bm I_C \otimes \bm h^{(1)} & \ldots & \bm 0_{MC\times C}\\
\cdots & \cdots & \cdots & \cdots \\
\bm 0_{MC\times C} & \bm 0_{MC\times C}& \ldots &  \bm I_C \otimes \bm h^{(Q-1)}
\end{array}
\right]\; .
\end{equation}
It is trivial to show that substituting the relation $\bm K'= \widetilde{\bm H} \widetilde{\bm H}^H$ into the RHS of  \eqref{eq:main_result} and letting $M$ diverge we obtain the same result as the one reported in \eqref{eq:limit001}.

Consider now the left-hand-side (LHS) of \eqref{eq:main_result}. 
This term, for finite $M$, can be written as
\[
 \log \det \left[\bm I_{N}+\underbrace{\widetilde{\bm H}^H\left(\sigma^2_w \bm I_{NM}+\bm K_{\rm c} \right)^{-1}\widetilde{\bm H}}_{\bm G} \right]\; ,
\]
where $\bm G$ is $(N \times N)$-dimensional. Letting $\bm K_c=\bm U_c \bm \Lambda_c \bm U_c^H $, with 
$\bm U_c \in \mathbb{C}^{NM \times L}$, $\bm \Lambda_c=\text{diag}(\lambda_{c,1}, \ldots , \lambda_{c,L})$, $L$ being the rank of $\bm K_c$, and $\bm U_c^H \bm U_c=\bm I_L$, applying the 
matrix inversion lemma we have, for $\bm G$:
\[
\begin{array}{lll}
\bm G=& \ds \frac{1}{\sigma^2_w}\widetilde{\bm H}^H \widetilde{\bm H}- \\ &
\ds \frac{1}{\sigma^2_w}\widetilde{\bm H}^H\bm U_c \text{diag}\left(\frac{\lambda_{c,1}}{\sigma^2_w+\lambda_{c,1}},\ldots , \frac{\lambda_{c,L}}{\sigma^2_w+\lambda_{c,L}} \right)\bm U_c^H \widetilde{\bm H}\; .
\end{array}
\]
Now we notice that
\[\begin{array}{lll}
\bm G=&M\left[\ds \frac{1}{M\sigma^2_w}\widetilde{\bm H}^H \widetilde{\bm H} - \right.\\ &
\left.\underbrace{\frac{1}{M \sigma^2_w}\widetilde{\bm H}^H\bm U_c \text{diag}\left(\frac{\lambda_{c,1}}{\sigma^2_w+\lambda_{c,1}},\ldots , \frac{\lambda_{c,L}}{\sigma^2_w+\lambda_{c,L}} \right)\bm U_c^H \widetilde{\bm H}}_{\bm D_M}\right]\; .
\end{array}\]
To complete the proof, we need to show that the (non-negative definite) matrix sequence $\bm D_M$, for large $M$, becomes small with respect to  $\frac{1}{M\sigma^2_w}\widetilde{\bm H}^H \widetilde{\bm H}$, which converges to
$\frac{p\beta}{\sigma^2_w} {\bm I}_{N}$. To show this, we will prove that     
\[
\lim_{M \rightarrow \infty} \bm D_M= \bm 0_{N \times N} \; ,
\]
with convergence in the mean square sense. 
To this end, we will show the following:
\begin{equation}
\lim_{M \rightarrow \infty} \mathbb{E} \left[ \text{tr}(\bm D_M)\right]=0\; ,  \quad \lim_{M \rightarrow \infty} \mathbb{E} \left[ \text{tr}^2(\bm D_M)\right]=0 \; .
\end{equation}
First of all,  consider that:
\[
\mathbb{E} \left[ \text{tr}(\bm D_M)\right]=\text{tr} \left[ \frac{1}{M\sigma^2_w}\bm U_c \bm \Lambda'_c \bm U_c^H \mathbb{E}\left[\widetilde{\bm H} \widetilde{\bm H}^H \right]\right]\; ,
\]
where $\bm \Lambda'_c=\text{diag}\left(\frac{\lambda_{c,1}}{\sigma^2_w+\lambda_{c,1}},\ldots , \frac{\lambda_{c,L}}{\sigma^2_w+\lambda_{c,L}} \right)$. Since
$\mathbb{E}\left[\widetilde{\bm H} \widetilde{\bm H}^H \right]=p \beta \bm I_{MN}$, we have
\[
\mathbb{E} \left[ \text{tr}(\bm D_M)\right]=\text{tr} \left[ \frac{p \beta}{M\sigma^2_w}\bm U_c \bm \Lambda'_c \bm U_c^H \right] \xrightarrow{M \rightarrow \infty}  0\, ,
\]
since $\text{tr}(\bm \Lambda'_c ) \leq L$. 
Moreover
\[
\begin{array}{lll}
\sigma^2_w\text{tr}^2[\bm D_M]= \text{tr}^2\left[\frac{1}{M}\widetilde{\bm H}^H \bm U_c \bm \Lambda'_c \bm U_c^H \widetilde{\bm H} \right]\leq \\
\text{tr}^2\left[\frac{1}{M}\widetilde{\bm H}^H \bm U_c \bm U_c^H \widetilde{\bm H} \right]\; ,
\end{array} \]
and
\[
\text{tr}^2\!\left[\!\frac{1}{M}\widetilde{\bm H}^H \bm U_c \bm U_c^H \widetilde{\bm H} \!\right]\!=\!\text{tr}^2\!\!\left[\underbrace{\frac{1}{\sqrt{M}}\widetilde{\bm H}^H \bm U_c}_{\bm A_M^H}
\underbrace{\frac{1}{\sqrt{M}} \bm U_c^H \widetilde{\bm H}}_{\bm A_M}\! \right]\, .
\]

Consider the $(L \times N)$-dimensional matrix $\bm A_M$, its generic $(\ell,cq)$-th entry can be written as
\begin{equation}
\left[\bm A_M\right]_{(\ell,cq)}=\frac{\sqrt{p}}{\sqrt{M}} \bm u_{\ell}^{(qC-C+c),H}\bm h^{(q-1)}\; ,
\end{equation}
$c=1,\ldots,C \; , q=1,\ldots,Q \; , \ell=1,\ldots, L$, where $\bm u_{\ell}^{(i)}$ is the $i$-th $M$-dimensional block of the $\ell$-th column of the matrix $\bm U_c $ and
 \begin{equation}
\ds\sum_{i=1}^{N}\parallel \bm u_{\ell}^{(i)} \parallel^2=1 \; , \ell=1,\ldots, L.
\end{equation}

\noindent
Notice that
\begin{equation}
\frac{\sqrt{p}}{\sqrt{M}} \mathbb{E} \left[ \bm u_{\ell}^{(qC-C+c),H}\bm h^{(q-1)}\right]=0\; ,
\end{equation}
while its mean square value is
\begin{equation}
\begin{array}{lll}
&\ds\frac{p}{M} \mathbb{E} \left[  \bm u_{\ell}^{(qC-C+c),H}\bm h^{(q-1)}\bm h^{(q-1),H} \bm u_{\ell}^{(qC-C+c)}\right] \\&=\ds\frac{p \beta}{M}\parallel \bm u_{\ell}^{(qC-C+c)}\parallel^2 \leq \ds\frac{p \beta}{M}\; ,
\end{array}
\end{equation}
where the fact that $\mathbb{E} \left[ \bm h^{(q-1)}\bm h^{(q-1),H}\right]=\beta \bm I_M$ has been exploited. As a consequence, we have:
\begin{equation}
\lim_{M \rightarrow \infty} \bm A_M = \bm 0 \,.
\end{equation}
Applying the continuous mapping theorem \cite{Convergence_Prob1999,shao2003mathematical}, we have that 
\[
\lim_{M \rightarrow \infty} \text{tr}^2 \left[ \bm A_M \bm A_M^H \right]=\text{tr}^2 \left[\lim_{M \rightarrow \infty}\bm A_M \bm A_M^H \right]=0\; ,
\]
which proves that 
$
\lim_{M \rightarrow \infty } \text{tr}^2 \left( \bm D_M \right)=0
$, 
and thus
\[
\lim_{M \rightarrow \infty} \bm G= \lim_{M \rightarrow \infty} \frac{1}{\sigma^2_w} \widetilde{\bm H}^H \widetilde{\bm H}= \ds \frac{p \beta}{\sigma^2_w}{\bm I}_N\; .
\]

\end{proof}

\section{Receiver processing} \label{receiver_processing_section}

In this section we focus on the signal processing algorithms at the BS to estimate the uplink channels and decode the MSs data symbols. We discuss both the case that the receiver has access to a clutter map, i.e. it has knowledge of the delays $\tau_q$ and directions of arrival $\theta_q$ of the clutter echoes, and the case that such information is not available.

\subsection{Uplink channel estimation} \label{CE_section}
We start considering the training phase, where the MSs send pilot signals to allow channel estimation at the BS. Given the data model \eqref{eq:observable_estimation_q}, we detail two different channel estimation strategies; the former does not need any information about  the clutter at the BS, while the latter assumes knowledge of the clutter statistics. We define
\begin{equation}
\widetilde{\mathbf{P}}_k^{(q)} \triangleq  \frac{\mathbf{P}_k^{(q)}}{\norm{\mathbf{P}_k^{(q)}}^2} \, \forall \; k = 1, \ldots, K \, , \forall \, q=1,\ldots, Q.
 \label{eq:P_tilde_def}
\end{equation}

\subsubsection{Pilot matched channel estimation (PM CE)} \label{PM_CE}
A simple estimator for the channel vector $\bh_k^{(q)}$, $\forall k, q$, is obtained through the following pilot-matched (PM) processing
\begin{equation}
\widehat{\bh}_k^{(q)}= \mathcal{Y}_q \ds\frac{\widetilde{\mathbf{P}}_k^{(q) \, *} }{\sqrt{p_{\rm{p},k}}} \; .
\label{eq:PM_estimation}
\end{equation}
The above processing only needs  knowledge of the normalized pilot sequences in \eqref{eq:P_tilde_def} and of the power transmitted by the users during the uplink training.

\subsubsection{MMSE channnel estimation (MMSE CE)} \label{LMMSE_CE}
Alternatively, a better performing estimator can be obtained by resorting to the linear MMSE criterion.
We focus on the observable in Eq. \eqref{eq:observable_estimation_q}.
The BS forms the following $M$-dimensional vector
\begin{equation}
\mathbf{r}_{q,k} = \mathcal{Y}_q \widetilde{\mathbf{P}}_k^{(q)\, *} =   \ds \sum_{j=1}^K\sqrt{p_{\rm{p},j}}\mathbf{h}_j^{(q)} \mathbf{P}_j^{(q) \, T} \widetilde{\mathbf{P}}_k^{(q)\, *}
+ \widetilde{\mathbf{w}}_{q,k} + \mathcal{C}_q \widetilde{\mathbf{P}}_k^{(q)\, *}\; ,
\label{eq:observable_estimation_q_k}
\end{equation}
where $\widetilde{\mathbf{w}}_{q,k}=\mathcal{W}_q \widetilde{\mathbf{P}}_k^{(q)\, *}$ is a $\mathcal{CN}\left(0, \sigma_w^2 \mathbf{I}_M\right)$ random vector.
The MMSE estimate of the $M-$dimensional channel vector $\mathbf{h}_k^{(q)}$ can be then computed as follows\cite{kay1998}:
\begin{equation}
\widehat{\mathbf{h}}_k^{(q)}= \mathbf{D}_{q,k} \mathbf{r}_{q,k} \, ,
\label{MMSE_estimate}
\end{equation}
where 
\begin{equation}
\mathbf{D}_{q,k}=\mathbb{E}\left[\mathbf{h}_k^{(q)} \mathbf{r}_{q,k}^H\right] \left( \mathbb{E}\left[\mathbf{r}_{q,k} \mathbf{r}_{q,k}^H\right]\right)^{-1}=\sqrt{p_{\rm{p},k}} \beta_k^2 \mathbf{R}_{q,k}^{-1},
\label{D_qk}
\end{equation}
with $\mathbf{R}_{q,k}\triangleq \mathbb{E}\left[\mathbf{r}_{q,k} \mathbf{r}_{q,k}^H\right]$ the covariance matrix of the vector  $\mathbf{r}_{q,k}$.  In order to provide an explicit expression for this matrix, we first  rewrite Eq. \eqref{eq:observable_estimation_q_k} as
\begin{equation}
\begin{array}{llll}
\mathbf{r}_{q,k} &= \ds \sum_{j=1}^K\sqrt{p_{\rm{p},j}}\mathbf{h}_j^{(q)} \mathbf{P}_j^{(q) \, T} \widetilde{\mathbf{P}}_k^{(q)\, *}
+ \widetilde{\mathbf{w}}_{q,k} \\ & + \left( \mathbf{I}_M \otimes \widetilde{\mathbf{P}}_k^{(q)\, H}\right) \mathbf{F} \text{vec} \left(\mathcal{C}_q \right) \; .
\end{array}
\label{eq:observable_estimation_q_k2}
\end{equation}
In \eqref{eq:observable_estimation_q_k2}, $\mathbf{F}$ is an $(MTC \times MTC)$-dimensional permutation matrix such that $\text{vec} \left(\mathcal{C}_q^T \right)=\mathbf{F} \text{vec} \left(\mathcal{C}_q \right)$, and the relation 
\begin{equation}
\mathcal{C}_q \widetilde{\mathbf{P}}_k^{(q)\, *}= \left( \mathbf{I}_M \otimes \widetilde{\mathbf{P}}_k^{(q)\, H}\right) \text{vec} \left(\mathcal{C}_q^T \right)  \, ,
\end{equation}
has been used. 

Given \eqref{eq:observable_estimation_q_k2}, 
it is straightforward to express 
 $\mathbf{R}_{q,k}$ as the superposition of the following three contributions:
\begin{equation}
\mathbf{R}_{q,k}=\ds \sum_{j=1}^K p_{\rm{p},j}\beta_j^2 \mathbf{I}_M \left|\mathbf{P}_j^{(q) \, T} \widetilde{\mathbf{P}}_k^{(q)\, *}\right|^2 + \sigma_w^2 \mathbf{I}_M +  \widetilde{\mathbf{K}}_{{\rm c},k}^{(q)} \, ,
\label{R_y}
\end{equation}
where
\begin{equation}
\!\!\widetilde{\mathbf{K}}_{{\rm c},k}^{(q)}\!=\!\left(\! \mathbf{I}_M \!\otimes\! \widetilde{\mathbf{P}}_k^{(q)\, H}\!\right) \!\mathbf{F} \mathbb{E}\!\!\left[ \!\text{vec} \left(\mathcal{C}_q \right)\text{vec} \left(\mathcal{C}_q \right)^H\right]\!\! \mathbf{F}^H \!\!\left( \!\mathbf{I}_M \!\otimes \! \widetilde{\mathbf{P}}_k^{(q)}\right) \, .
\label{K_c_q_tilde}
\end{equation}
To fully specify  
$\mathbf{R}_{q,k}$, we still need to provide an explicit expression for  $\mathbb{E}\left[ \text{vec} \left(\mathcal{C}_q \right)\text{vec} \left(\mathcal{C}_q \right)^H\right]$. To this end, we assume that, after the CP cancellation stage, the set of the scatterers corrupting the reception of different packets are disjoint, i.e. $ {\cal S}(\ell_1) \cap {\cal S}(\ell_2) = \emptyset$. Under these circumstances and using the expressions in Eqs. \eqref{clutter_matrix} and \eqref{R_tilde_definition}, it is easily shown that
\begin{equation}
\mathbb{E}\!\!\left[ \!\text{vec} \left(\mathcal{C}_q \right)\text{vec} \left(\mathcal{C}_q \right)^H\!\right]=\mbox{blkdiag}\left(\bm K_{\rm c}\left(1\right)^{(q)}, \ldots,
\bm K_{\rm c}\left(T\right)^{(q)}\right)\; ,
\end{equation}
where $\bm K_{\rm c}\left(\ell\right)^{(q)}$ is the following  $(MC \times MC)-$ dimensional matrix
\begin{equation}
\begin{array}{lll}
\bm K_{\rm c}\left(\ell\right)^{(q)}= & \ds \sum_{p \in {\cal S}(\ell)}
\sum_{m=0}^{Q-1} P_T \sigma^2_{p,m} \left[\widetilde{\bR}_{p,\ell,m}((q-1)C+1:qC)
\right. \\ & \left.
\widetilde{\bR}_{p,\ell,m}^H((q-1)C+1:qC) \right] \otimes \mathbf{b}(\theta_p)\mathbf{b}^H(\theta_p) \; ,
\end{array}
\end{equation}
with $\ell=1, \ldots, T$ and  $\sigma^2_{p,m}= \mathbb{E} \left[ |\beta_{p,m}|^2\right]$.

As a final remark, we notice that the MMSE channel estimation procedure, unlike the PM estimator, assumes complete knowledge of the clutter statistics and of the noise variance. This assumption, along with the heavier computational complexity entailed by matrix inversion, is expectedly rewarded by increased robustness to the clutter disturbance {\em and} to the multiuser interference, inherent in the considered non-orthogonal multiple access scheme.

\subsection{Uplink data detection} \label{Detection_strategies}
We now focus on the problem of uplink data detection, processing separately, for the sake of simplicity, the columns of the received matrix $\bY(\ell)$; recall that the $n$-th column of $\bY(\ell)$, 
$\mathbf{y}^{(n)}(\ell)$ is expressed as in \eqref{eq:vector_observable}.  As anticipated, we consider both the case of complete prior knowledge of clutter covariance properties, and the case that the relevant clutter parameters - such as the angles at which scattering centers are located - are unknown and must be averaged out.
Our baseline detector is the classical  \textit{channel-matched beamforming (CM)}. Based on the channel estimate $\widehat{\bh}_k^{(\lceil n/C \rceil)}$, a soft estimate of $\mathbf{X}_k(\ell)^{(n)}$ is built as
\begin{equation}
\widehat{\mathbf{X}}_k(\ell)^{(n)}=\ds \frac{\widehat{\bh}_k^{(\lceil n/C \rceil)\, H}
\by^{(n)}(\ell) }{\sqrt{p_k}\left\|\widehat{\bh}_k^{(\lceil n/C \rceil)}\right\|^2} \; .
\end{equation}
This detector just relies on the $M$-dimensional channel signature to reject the clutter and multiuser interference. 

\subsubsection{Clutter aware processing (CAP)} \label{CAP_section}
In this scenario, the covariance matrix of the clutter vector $\mathbf{C}(\ell)^{(n)}$ is assumed known, and takes on the form:
\begin{equation}
\mathbf{K}_{\mathbf{C}(\ell)^{(n)}}=\ds \sum_{q \in {\cal S}(\ell)}
\sum_{m=0}^{Q-1}P_T  \sigma^2_{q,m} \left|\widetilde{\bR}_{q,\ell,m}^{(n)}\right|^2
\mathbf{b}(\theta_q)\mathbf{b}^H(\theta_q) \; .
\label{clutter_covariance_ell_n}
\end{equation}
A number of linear receivers exploiting such a knowledge to demodulate the data symbols $\mathbf{X}_k(\ell)^{(n)}$ based on the data model
\eqref{eq:vector_observable} can thus be implemented. 

\begin{itemize}
\item \textit{Zero-Forced clutter (ZF)}. This receiver exploits the low-rank feature of the clutter covariance matrix  \eqref{clutter_covariance_ell_n} and  zero-forces the clutter contribution by projecting the
observable data vector along a direction that is orthogonal to the clutter subspace. 
Letting $\mathbf{U}(\ell)^{(n)}$ be a matrix containing the eigenvectors of the matrix $\mathbf{K}_{\mathbf{C}(\ell)^{(n)}}$ associated to non-zero eigenvalues, we have in this case
\begin{equation}
\widehat{\mathbf{X}}_k(\ell)^{(n)}\!=\!\ds \frac{\left[
\left( \bI_{M}\!-\! \mathbf{U}(\ell)^{(n)}\mathbf{U}(\ell)^{(n)\, H}\right)
\widehat{\bh}_k^{(\lceil n/C \rceil)}\right]^H \!\!\!
\by^{(n)}(\ell) }{\sqrt{p_k}\left\|\left( \bI_{M}- \mathbf{U}(\ell)^{(n)}\mathbf{U}(\ell)^{(n)\, H}\right)\widehat{\bh}_k^{(\lceil n/C \rceil)}\right\|^2}  .
\label{eq:ZFprocessing}
\end{equation}

\item \textit{Linear MMSE data detector}. 
This receiver performs a linear MMSE estimation of the data symbol ${\mathbf{X}}_k(\ell)^{(n)}$.
We thus have:
\begin{equation}
\widehat{\mathbf{X}}_k(\ell)^{(n)}= \sqrt{p_k} \widehat{\bh}_k^{(\lceil n/C \rceil)\, H}
\mathbf{K}_{\mathbf{y}(\ell)^{(n)}}^{-1}  \by^{(n)}(\ell) \; ,
\end{equation}
with
\begin{equation}
\mathbf{K}_{\mathbf{y}(\ell)^{(n)}} =\ds \sum_{j=1}^K p_j
\widehat{\bh}_j^{(\lceil n/C \rceil)}  \widehat{\bh}_j^{(\lceil n/C \rceil)\, H} +
\sigma^2_w \bI_M+ \mathbf{K}_{\mathbf{C}(\ell)^{(n)}} \; .
\end{equation}
The LMMSE data detector provides robustness against both the clutter disturbance and the multiuser interference by the other users in the system. 

\item \textit{Full zero-forcing (FZF)}: 
The above receiver is capable of nulling the clutter contribution, but does not provide any improved protection with respect to the multiuser interference. To circumvent this drawback, we thus consider 
an FZF receiver that zero-forces both disturbance sources. Let thus $\mathbf{U}_k(\ell)^{(n)}$ be a matrix containing the eigenvectors of the matrix
$\mathbf{K}_{\mathbf{y}(\ell)^{(n)}} - \sigma^2_w \bI_M - p_k
\widehat{\bh}_k^{(\lceil n/C \rceil)}  \widehat{\bh}_k^{(\lceil n/C \rceil)\, H}$
associated to non-zero eigenvalues. The data estimator of   $\mathbf{X}_k(\ell)^{(n)}$ is now written as
\begin{equation}
\widehat{\mathbf{X}}_k(\ell)^{(n)}\!=\!\ds \frac{\left[
\left( \bI_{M}\!\!-\!\! \mathbf{U}_k(\ell)^{(n)}\mathbf{U}_k(\ell)^{(n)\, H}\right)
\widehat{\bh}_k^{(\lceil n/C \rceil)}\right]^H\!\!\!
\by^{(n)}(\ell) }{\sqrt{p_k}\left\|\left( \bI_{M}- \mathbf{U}_k(\ell)^{(n)}\mathbf{U}_k(\ell)^{(n)\, H}\right)\widehat{\bh}_k^{(\lceil n/C \rceil)}\right\|^2}  .
\end{equation}
\end{itemize}

\subsubsection{No Clutter aware processing (NCAP)} \label{NCAP_section}
The above CAP receivers assume knowledge of the clutter covariance matrix, and in particular of the clutter direction of arrivals $\{\theta_q\}$. We now detail two receiver structures that 
do not rely on this information. 

\begin{itemize}

\item \textit{Bessel-based zero-forced clutter (BZF)}. 
One possible way to avoid relying on the knowledge of the clutter direction-of-arrival angles is to model them as independent random variables, uniformly distributed  on $[-\pi, \pi]$. The clutter covariance matrix in this case is obtained by averaging \eqref{clutter_covariance_ell_n} with respect to the angles $\{\theta_q\}$. Letting 
\begin{equation}
\mathbf{B}= \mathbb{E}\left[\mathbf{b}(\theta)\mathbf{b}^H(\theta)\right],
\end{equation}
and exploiting the definition of $\mathbf{b}(\theta)$ in Eq. \eqref{BS_ULA}, we have that the generic entry of the matrix $\mathbf{B}$ can be evaluated as follows:
\begin{equation}
\begin{array}{lll}
\left[\mathbf{B}\right]_{(\ell,m)}&= \mathbb{E}\left[e^{-j2\pi \frac{d}{\lambda} \sin\left(\theta\right) \left( \ell - m \right)}\right]\\ & \ds =\frac{1}{2\pi} \int_{- \pi}^{\pi} {e^{-j2\pi \frac{d}{\lambda} \sin\left(\theta\right) \left( \ell - m \right)} \, d \theta}\\ & \ds= J_0\left(  \frac{2\pi d}{\lambda} \left( \ell - m \right) \right) \; ,
\end{array}
\end{equation}
where $J_0(x)=\frac{1}{2\pi} \displaystyle{ \int_{- \pi}^{\pi} {e^{-jx\sin\left( \theta \right)} \, d\theta}}$ is the first kind Bessel function of order 0. 
It is easy to realize that the clutter covariance matrix averaged with respect to the direction-of-arrival is proportional to $\mathbf{B}$. A possible detection strategy is thus to use the beamformer reported in \eqref{eq:ZFprocessing}, where now the matrix $\mathbf{U}(\ell)^{(n)}$ is no longer dependent on $\ell$ and contains the eigenvectors associated to the $P$ largest eigenvalues of the matrix  $\mathbf{B}$, with $P$ a design parameter to be carefully tuned: needless to say, a sensible choice should compromise between the conflicting requirements of rejecting as much clutter as possible, while limiting the inevitable noise enhancement entailed by the projection operation.

\item \textit{Angles of arrival (AoA) estimation-based zero-forced clutter (AEZF)}. 
This strategy relies on the fact that the clutter disturbance appears in the observable vector at the BS with a signature with known functional form, i.e. the BS antenna array response at the unknown AoA's. 
In the following, we propose a heuristic procedure for estimating the clutter AoAs, and then consider a reciver that zero-forces the array response vectors corresponding to the estimated AoA's. 
We consider again the data model  \eqref{eq:vector_observable} corresponding to the transmission of  $N$ symbols, and consider the function
\begin{equation}
f_{\ell}\left(\theta\right)=\sum_{n=1}^N {\left|\mathbf{b}^H\left(\theta\right)\mathbf{y}(\ell)^{(n)}\right|}.
\label{function_estimation}
\end{equation}
Evidently, this function should exhibit some local maxima when $\theta$ approaches any of the clutter AoA's, and, in particular, the value of the maximum is an indicator of how strong is the clutter along that AoA. Accordingly, a possible detection strategy is to estimate the strongest clutter AoA's and to zero-force the steering vector associated to these angles.
In order to do so, we let  $\theta^{(r)}=\frac{\pi(r-1)}{R}-\pi,  \; \;  r=1,\ldots,R$ denote a set of $R>>1$ angles uniformly spanning the range $\left[ -\frac{\pi}{2},\frac{\pi}{2}\right]$, and  then  evaluate the mean and the variance of the function $f_{\ell}\left(\theta\right)$  evaluated on this angle grid, i.e. we have 
\begin{equation}
m_{\ell}=\frac{1}{R} \sum_{r=1}^R{ f_{\ell}\left(\theta^{(r)}\right)} , \; \; v_{\ell}=\frac{1}{R} \sum_{r=1}^R{\left( f_{\ell}\left(\theta^{(r)}\right)- m_{\ell} \right)^2}.
\label{mean_variance_function_est}
\end{equation}
We decide to zero-force all the steering vectors corresponding to the angles that belong to the following set:
\begin{equation}
\mathcal{R}_{\ell}= \lbrace \theta^{(r)} \, : \, f_{\ell}\left(\theta^{(r)}\right) \geq m_{\ell}+2\sqrt{v_{\ell}} \rbrace.
\end{equation} 
The threshold setting at the level $m_{\ell}+2\sqrt{v_{\ell}}$ is heuristic and different choices can be obviously made. which, of course, have an impact on the cardinality of the set 
$\mathcal{R}_{\ell}$, and, thus, on the number of zero-forced directions of arrival.  
Additionally, the performance of this strategy also depends on how dense is the angle grid (i.e. on how large is $R$); in general, the larger $R$, the better the accuracy in the estimation of the clutter AoA's. 
\end{itemize}

\section{Uplink spectral efficiency derivation}

\begin{figure*}
\begin{equation}
\begin{array}{ll}
\widetilde{\mbox{SINR}}_k(\ell)^{(n)}= \ds
\frac{p_k |\mathbf{v}_k^{(q) \,  H} \mathbf{h}_k^{(q)}|^2}
{\ds \sum_{j \neq k}p_j |\mathbf{v}_k^{(q) \,  H} \mathbf{h}_j^{(q)}|^2
+ \sigma^2_w \|\mathbf{v}_k^{(q)}\|^2 + \mathbf{v}_k^{(q) \,  H} 
\mathbf{K}_{\mathbf{C}(\ell)^{(n)}}
\mathbf{v}_k^{(q) } 
} \; .
\end{array}
\label{eq:SINRtilde}
\end{equation}
\hrulefill
\end{figure*}

All of the previously outlined data detection strategies are linear. Denoting by $\mathbf{v}_k^{(q)}$, with $q=\lceil n/C \rceil$, 
 the vector used for detecting the data symbols $\mathbf{X}_k(\ell)^{(n)}$, based on the observable $\mathbf{y}(\ell)^{(n)}$, it is easily shown that the  post-detection SINR relative to user $k$,
$\widetilde{\mbox{SINR}}_k(\ell)^{(n)}$ say, can be expressed as is Eq. \eqref{eq:SINRtilde} at the top of next page. Although \eqref{eq:SINRtilde} provides a reasonable and correct expression for the SINR, it can be used to compute the achievable spectral efficiency (SE) through the Shannon rate formula only in the case in which the receiver has perfect CSI. Indeed, notice that \eqref{eq:SINRtilde} depends on the channel coefficients and on the beamformers, that, in turn, also depend on the channel coefficients. In case of perfect CSI, this SINR expression can be plugged into the Shannon rate formula in order to have the system achievable SE. Things are instead different when perfect CSI is not available. In this case, the SINR expression \eqref{eq:SINRtilde}  still holds, but it contains now the true channel values and the "approximate" beamformers, that have been computed based on the noisy channel estimates. Since \eqref{eq:SINRtilde} contains the true channel coefficients, that are unknown to the receiver, it follows that the SINR expression cannot be exactly computed at the receiver, and, thus, the Shannon SE $\log_2(1+ \widetilde{\mbox{SINR}}_k(\ell)^{(n)})$ 
is no longer "attainable" and becomes an upper bound \cite{Caire_bounds_2018}. When the receiver has imperfect CSI, an effective performance in terms of SE can be obtained using only information that is available at the receiver. In order to be able to analyze the system in terms of uplink SE, in the following we provide a lower bound  in the case of PM and MMSE CE, CM detection and knowledge of clutter covariance matrix at the BS. The bounding technique exploits the channel estimates only for computing the receive combining vectors, {while this information is not exploited
in the signal detection phase.} This simplification is reasonable when there is
substantial \textit{channel hardening}\cite{marzetta2016fundamentals,MassiveMIMO_book_Bjornson}.
Our analysis is carried on with reference to the $n$-th subcarrier of the $k$-th user in a generic packet;  in order to simplify the notation, we omit the packet index.
Denoting by  $\mathbf{v}_k^{(q)}$, with $q=\lceil n/C \rceil $,  the combining vector for the data transmitted by $k$-th user, and using Eq. \eqref{eq:vector_observable}, we have:

\begin{equation}
\begin{array}{lll}
\widehat{\mathbf{X}}_k^{(n)}&=\mathbf{v}_k^{(q)\, H}\mathbf{y}^{(n)}= \sqrt{p_k} \mathbf{X}_k^{(n)} \mathbf{v}_k^{(q)\, H} \mathbf{h}_k^{(q)} \\ & + \ds \sum_{\substack{j=1 \\ j \neq k}}^K
\sqrt{p_j}  \mathbf{v}_k^{(q)\, H} \mathbf{h}_j^{(q)} \mathbf{X}_j^{(n)}+ \mathbf{v}_k^{(q)\, H} \mathbf{W}^{(n)} + \mathbf{v}_k^{(q)\, H} \mathbf{C}^{(n)} \; ,
\end{array}
\label{eq:estimates_k_n}
\end{equation}
By adding and subtracting $\sqrt{p_k} \mathbb{E}\left[\mathbf{v}_k^{(q)\, H} \mathbf{h}_k^{(q)}\right] \mathbf{X}_k^{(n)}$, the signal in  \eqref{eq:estimates_k_n} can be rewritten as
\begin{equation}
\begin{array}{lll}
\widehat{\mathbf{X}}_k^{(n)}=&\underbrace{\sqrt{p_k} \mathbb{E}\left[\mathbf{v}_k^{(q)\, H} \mathbf{h}_k^{(q)}\right]\mathbf{X}_k^{(n)}}_{\text{Desired signal over average channel}}\\ & + \underbrace{\sqrt{p_k} \left( \mathbf{v}_k^{(q)\, H} \mathbf{h}_k^{(q)} - \mathbb{E}\left[\mathbf{v}_k^{(q)\, H} \mathbf{h}_k^{(q)}\right]\right) \mathbf{X}_k^{(n)}}_{\text{Desired signal over ``unknown'' channel}} \\ & + \underbrace{\ds \sum_{\substack{j=1 \\ j \neq k}}^K
\sqrt{p_j} \mathbf{v}_k^{(q)\, H} \mathbf{h}_j^{(q)} \mathbf{X}_j^{(n)} }_{\text{Interference}}+ \underbrace{\mathbf{v}_k^{(q)\, H} \mathbf{W}^{(n)}}_{\text{Noise}} + \underbrace{\mathbf{v}_k^{(q)\, H} \mathbf{C}^{(n)}}_{\text{Clutter disturbance}} \; ,
\end{array}
\label{eq:estimates_k_n2}
\end{equation}
 Only the part of the desired signal received over the average precoded channel $ \mathbb{E}\left[\mathbf{v}_k^{(q)\, H} \mathbf{h}_k^{(q)}\right]$ is treated
as the true desired signal. The part of $\mathbf{X}_k^{(n)}$ received over the deviation
from the mean value $ \mathbf{v}_k^{(q)\, H} \mathbf{h}_k^{(q)} - \mathbb{E}\left[\mathbf{v}_k^{(q)\, H} \mathbf{h}_k^{(q)}\right]$ has zero mean and can thus
be treated as an uncorrelated noise signal in the detection phase. The UL ergodic channel capacity of  the $k-$th user on the $n-$th subcarrier is thus lower bounded by

\begin{equation}
\text{SE}_k^{(n)}= \frac{N_{\rm pkt}-T}{N_{\rm pkt}}\log_2 \left(1 + \text{SINR}_k^{(n)}\right) \; \text{[bit/s/Hz]} \, ,
\label{SE_lower_bound}
\end{equation}
with  $\text{SINR}_k^{(n)}$ shown in \eqref{SINR_general} at the top of the next page.

\begin{figure*}[t]
\begin{equation}
\text{SINR}_k^{(n)}= \frac{p_k \left| \mathbb{E}\left[\mathbf{v}_k^{(q)\, H} \mathbf{h}_k^{(q)}\right]\right|^2}{\ds \sum_{\substack{j=1}}^K
p_j \mathbb{E} \left[ \left|\mathbf{v}_k^{(q)\, H} \mathbf{h}_j^{(q)}\right|^2\right]- p_k \left| \mathbb{E}\left[\mathbf{v}_k^{(q)\, H} \mathbf{h}_k^{(q)}\right]\right|^2 + \sigma^2_w \mathbb{E}\left[ \norm{\mathbf{v}_k^{(q)}}^2 \right] + \mathbf{E} \left[ \left| \mathbf{v}_k^{(q)\, H} \mathbf{C}^{(n)} \right|^2\right]}
\label{SINR_general}
\end{equation}
\end{figure*}
The lower bound in Eq. \eqref{SE_lower_bound} is known as the use-and-then-forget (UatF) bound since the channel estimates are used for combining and then effectively ``forgotten'' before signal detection \cite{marzetta2016fundamentals}. Note that the SINR expression in Eq. \eqref{SINR_general} at the top of next page is deterministic and contains several expectations over the random channel realizations. 
For the case that an arbitrary combining vector is used, 
each of these expectations can be individually computed by means of Monte Carlo simulation. For CM combining,
instead, they can be obtained in closed form. 
Indeed, it can be shown that with CM combining and  PM channel estimation, Eq. \eqref{SINR_general} can be written as in \eqref{SINR_MR_PM} at the top of next page,
\begin{figure*}
\begin{equation}
\text{SINR}_{k, {\rm PM}}^{(q)}= \frac{p_k p_{\rm{p},k} \beta_k^4 M^2}{
\begin{array}{lll}
&\ds \sum_{\substack{j=1 \\ j\neq k}}^K
p_j  p_{\rm{p},j} \beta_j^4 M^2 \left| \mathbf{P}_j^{(q) \, T} \widetilde{\mathbf{P}}_k^{(q)\, *}\right|^2 + \ds \sum_{\substack{j=1}}^K p_j \beta_j^2  \text{tr} \left( \mathbf{R}_{q,k}\right) +  \sigma^2_w  \text{tr}\left( \mathbf{R}_{q,k} \right) + \text{tr}\left( \mathbf{R}_{q,k} \mathbf{K}_{\mathbf{C}^{(n)}} \right)
\end{array}}
\label{SINR_MR_PM}
\end{equation}
\end{figure*}
while, for the case of MMSE channel estimation, \eqref{SINR_general} is expressed as in
\eqref{SINR_MR_MMSE}, again shown on  next page.
\begin{figure*}\begin{equation}
\text{SINR}_{k, {\rm LMMSE}}^{(q)}= \frac{p_k p_{\rm{p},k} \beta_k^4 \text{tr}\left( \mathbf{D}_{q,k} \right)^2}{
\begin{array}{lll}
&\ds \sum_{\substack{j=1 \\ j\neq k}}^K
p_j  p_{\rm{p},j} \beta_j^4 \text{tr} \left( \mathbf{D}_{q,k}\right)^2 \left| \mathbf{P}_j^{(q) \, T} \widetilde{\mathbf{P}}_k^{(q)\, *}\right|^2 + \ds \sum_{\substack{j=1}}^K p_j \sqrt{p_{\rm{p},k}}\beta_j^2 \beta_k^2 \text{tr} \left( \mathbf{D}_{q,k}\right) \\& + \sqrt{p_{\rm{p},k}} \beta_k^2 \left[\sigma^2_w  \text{tr}\left( \mathbf{D}_{q,k} \right) + \text{tr}\left( \mathbf{D}_{q,k} \mathbf{K}_{\mathbf{C}^{(n)}} \right)\right]
\end{array}}
\label{SINR_MR_MMSE}
\end{equation}
\hrulefill
\end{figure*}
The proof of the validity of \eqref{SINR_MR_PM} and \eqref{SINR_MR_MMSE} is provided in Appendixes \ref{App_A} and  \ref{App_B}, respectively. As a final remark, we notice that the last summand in the denominator of  \eqref{SINR_MR_PM} and \eqref{SINR_MR_MMSE}  represents the effect of the clutter disturbance on the uplink SINR of the communication system. Although numerical results show that the overall system performance improves with increasing number of antennas $M$, it appears that the effect of this term does not vanish in the limit of large number of antennas.

\begin{figure*}
\begin{center}
\includegraphics[scale=0.43]{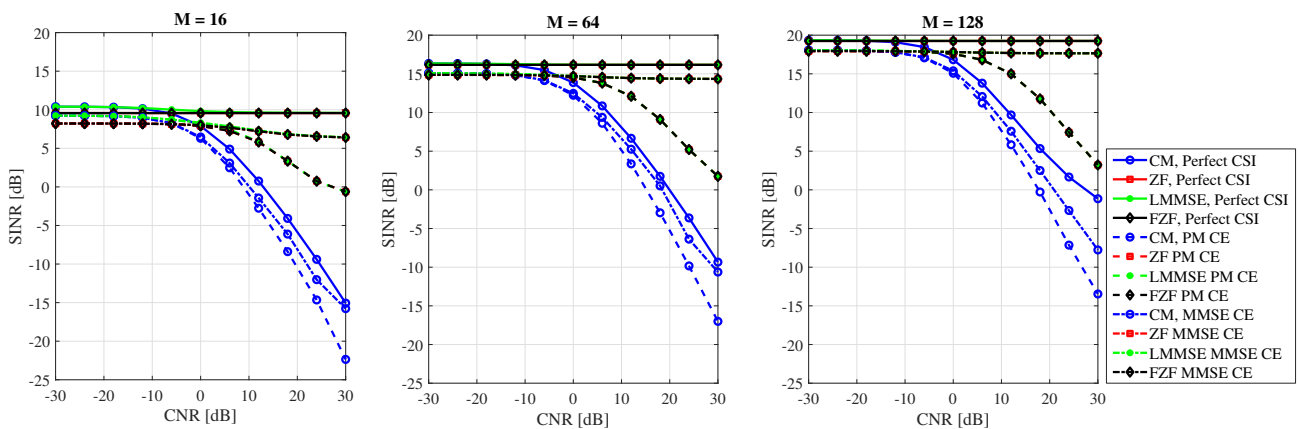}
\caption{SINR versus CNR of four detection strategies in the cases of perfect CSI, PM CE and MMSE CE, with $K=1$ and different values of $M$.}
\label{Fig:SINR_CE_single_user}
\end{center}
\end{figure*}

\section{Numerical results}
Numerical results are now shown in order to corroborate previous analytical findings and to illustrate the performance of the proposed detection structures. 
The simulation environment adopts the system parameters detailed in Section \ref{system_model_section}; additionally, the channel vectors between the BS and the MSs are generated through the superposition of small-scale Rayleigh distributed fading (independent across antennas at the BS large array), log-normal shadowing, and distance-dependent path-loss --  the three slope path loss model detailed in reference \cite{buzzi_CFUC2017} is used. The MSs distance from the BS is uniform in the range $[20, 500]$ m., the additive thermal noise is assumed to have a power spectral density of -174 dBm/Hz, and the front-end receiver is assumed to have a noise figure of 3 dB. The simulation parameters are summarized in Table \ref{sim_par_table}. 

\begin{table}[]
\centering
\caption{Simulation Parameters}
\label{sim_par_table}
\begin{tabular}{|p{0.7cm}|p{1.8cm}|p{4.3cm}|}
\hline
\textbf{Name}                   & \textbf{Value} & \textbf{Description}                                                                                                                \\ \hline
$f_c$                           & 3 GHz          & carrier frequency                                                                                                                   \\ \hline
$M$                             & 16, 64, 128    & number of antennas at the BS
 \\ \hline
$d$                             & $\frac{\lambda}{2}$    & antenna spacing
 \\ \hline
$K$                             & 1, 10, 30           & number of users in the cellular system uniformly distributed in the range {[}20,500{]} m\\ \hline
$p_k$                           & 100 mW         & MSs transmit power in training and data transmission phases                             \\ \hline
$N$                             & 4096           & number of subcarriers                                                                                                                \\ \hline
$\Delta_f$                      & 30 kHz         & subcarrier spacing                                                                                                                  \\ \hline
$C$                             & 16             & number of consecutive subcarriers where the channel is considered constant             
\\ \hline
$Q$                    & $\frac{N}{C}=256$            & number of scalar coefficients representing the amount of CSI for each user and for each BS antennas   
\\ \hline
$N_s$                    & 100            & number of total scatterers in the system uniformly distributed   in the range {[}1,150{]} km   
\\ \hline
$N_{\rm CP}$                    & 288            & discrete length of the cyclic-prefix                                                                                                \\ \hline
$T_{\rm PRT}$ & 1 ms           & radar pulse repetition time                                                                                                         \\ \hline
$N_{\rm pkt}$ & 14             & number of packets into a 0.5 ms timeslot                                                                                            \\ \hline
$T_s$                           & 8.146 ns       & symbol time                                                                                                                         \\ \hline
$T$                             & 7              & number of packets used for the channel estimation                                                                                   \\ \hline
$L$                             & 32             & discrete length of the radar coded waveform                                                                                         \\ \hline
$F$                             & 3 dB           & noise figure at the receiver                                                                                                        \\ \hline
$\mathcal{N}_0$                 & -174 dBm/Hz    & power spectral density of the noise                                                                                                 \\ \hline
\end{tabular}
\end{table}

We report results for the case of perfect CSI, for the case of  PM channel estimation as detailed in Section \ref{PM_CE} and for the case of MMSE channel estimation as detailed in Section \ref{LMMSE_CE}, with $T=7$.
The MSs transmit power is set at $100$ mW, both in the training and data transmission phases, i.e. $p_k=p_{{\rm p},k}=100$ mW, $\forall k=1,l\ldots,K$.

\begin{figure*}
\begin{center}
\includegraphics[scale=0.43]{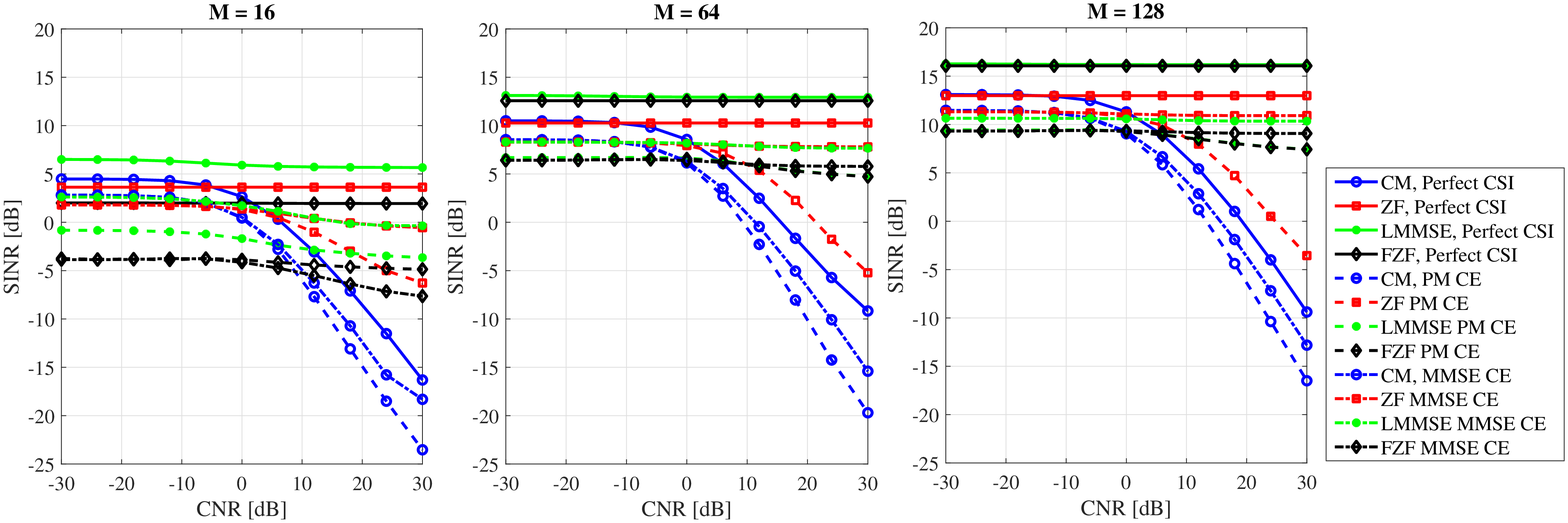}
\caption{SINR versus CNR of four detection strategies in the cases of perfect CSI, PM CE and MMSE CE, with $K=10$ and different values of $M$.}
\label{Fig:SINR_CE_multiple_users}
\end{center}
\end{figure*}

Fig. \ref{Fig:SINR_CE_single_user} reports the average per user SINR \eqref{eq:SINRtilde}  versus the Clutter-to-Noise Ratio (CNR) for the four detection strategies in the case of CAP detailed in Section \ref{CAP_section}, considering the case of perfect CSI, PM channel estimation and MMSE channel estimation; the figure considers a single-user system, and contains three subplots, in order to
illustrate the detectors' performance for three different values of the BS array size $M$.
Fig. \ref{Fig:SINR_CE_multiple_users} shows exactly the same results as Fig. \ref{Fig:SINR_CE_single_user}, with the only difference that a  multiuser system, with $K=10$ users, 
has been considered. Inspecting the figures, 
several comments can be made. First of all, results clearly show that, regardless of the data detection structure, performance steadily improves with increasing $M$. This provides a first numerical confirmation that massive MIMO systems are intrinsically resistant to co-existing radar clutter interference. The figures also show that, as expected, CM beamforming is the most vulnerable combining scheme to interference, while the other strategies exhibit much better performance. Sorting the detection strategies in ascending performance order we have CM, ZF, FZF and, finally, MMSE\footnote{Note that in the single-user case, ZF and FZF strategies end up coincident, and this is why the curves showing the performance of these two detectors in Fig. \ref{Fig:SINR_CE_single_user} perfectly overlap.}. Regarding channel estimation, as expected, MMSE channel estimation achieves better performance than PM channel estimation. Overall, Figs.  \ref{Fig:SINR_CE_single_user} and  \ref{Fig:SINR_CE_multiple_users} provide evidence that, if clutter second-order statistics are known at the BS, several strategies exist to tackle the additional interference caused by a co-existing radar system, with better and better performance for increasing number of BS antennas: increasing the antenna array size from $M=16$ to $M=128$ provides indeed a $10 \,$dB increase in the received SINR. 

\begin{figure}
\begin{center}
\includegraphics[scale=0.55]{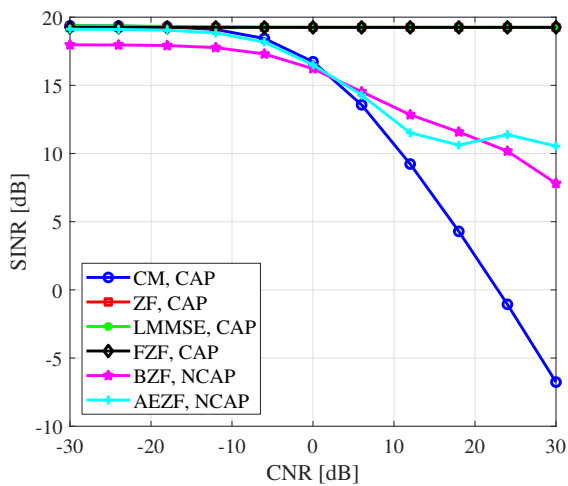}
\caption{SINR versus CNR of the proposed detection strategies in the case of CAP and NCAP, with $K=1$ and $M=128$.}
\label{Fig:SINR_NCAP_single_user}
\end{center}
\end{figure}

\begin{figure}
\begin{center}
\includegraphics[scale=0.55]{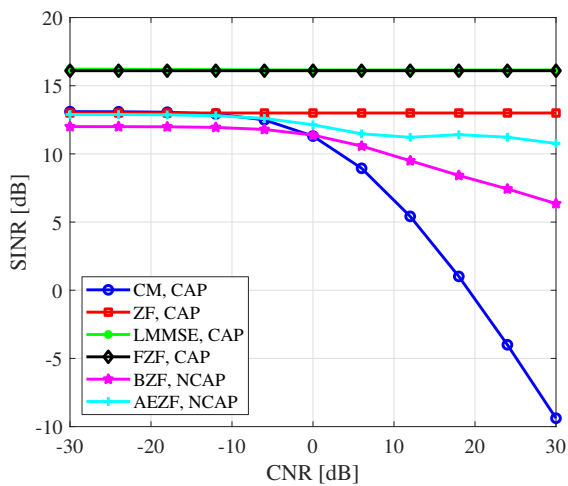}
\caption{SINR versus CNR of the proposed detection strategies in the case of CAP and NCAP, with $K=10$ and $M=128$.}
\label{Fig:SINR_NCAP_multiple_users}
\end{center}
\end{figure}

We now turn our attention to the performance of the receivers designed with no prior information on the clutter. 
Figs. \ref{Fig:SINR_NCAP_single_user} and \ref{Fig:SINR_NCAP_multiple_users} report the average SINR per user versus the CNR in the case of NCAP; for benchmarking purposes, we show in the same plot also the performance of the CM, ZF and FZF CAP rules. A BS array size $M=128$ and perfect CSI is assumed here. 
While Fig. \ref{Fig:SINR_NCAP_single_user} refers to a single-user system,  Fig. \ref{Fig:SINR_NCAP_multiple_users} refers instead to a multiuser system with $K=10$.
For the BZF processing, we use $P=35$, i.e. the ZF processing nulls the interference lying in the subspace spanned  by the eigenvectors of the matrix $\mathbf{B}$ corresponding to the 35 largest eigenvalues\footnote{With this choice  80$\%$ of the total energy of the matrix eigenvalues is captured.}. For the AEZF processing we use $R=500$, i.e., in order to estimate the AoA's we implement an exhaustive search over $500$ angles uniformly spanning the range $\left[ -\frac{\pi}{2},\frac{\pi}{2}\right]$.
From the figures it is seen that, even though ZF and FZF CAP rules are of course the best strategies and exhibit a behavior that is independent of the CNR, the proposed NCAP strategies outperform, in the critical region of large CNR, the CM beamformer. In particular, in the critical scenario that CNR$=30\,$ dB, the BZF and AEZF rules achieve in the multiuser case a SINR of 4 and 10.5 dB, respectively, while the CM beamformer achieves a SINR equal to $-8\,$dB. Moreover, it is seen that the AEZF performance exhibits a floor in the large CNR region, which seems reasonable since in this case the AoA's of the clutter disturbance are well estimated and their effect can be perfectly 
zero-forced.

\begin{figure*}
\centering
\includegraphics[scale=0.43]{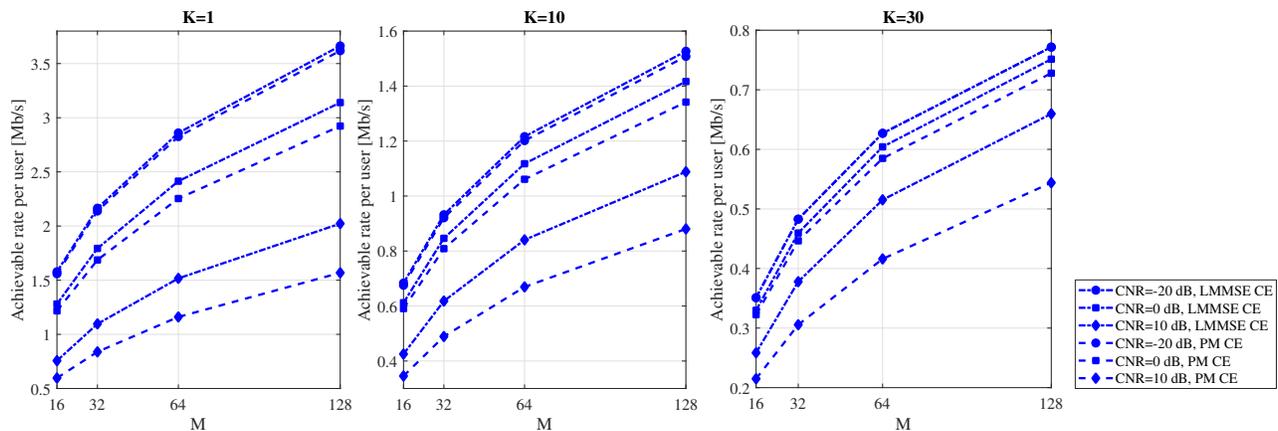}
\caption{Spectral efficiency lower bound versus $M$ for PM and MMSE channel estimation procedures for different values of CNR with $K=1$,  $K=10$, and $K=30$.}
\label{Fig:SE_complete}
\end{figure*}

Finally, Fig. \ref{Fig:SE_complete} reports the SE lower bounds
\eqref{SINR_MR_PM} and \eqref{SINR_MR_MMSE}, for the case of CM detection, and with PM and   
MMSE channel estimation. The results are plotted versus the number of antennas at the BS $M$, for three different values of the CNR; each subfigure refers to different numbers of users in the system. In particular, the third subfigure considers 
a quite overloaded system with $K=30$ users using the same frequency band. Results show that for increasing
number of users the achievable rate per user decreases, even though the overall sum-rate gets increased. This is in line
with the results reported in the massive MIMO literature.  
The results also confirm that the performance grows with increasing $M$, regardless of the value of the CNR. As expected, also in this case MMSE channel estimation permits attaining better performance than PM channel estimation.

\section{Conclusions}
The paper has considered a single-cell massive MIMO communication system using SC-FDMA multiple access in the uplink and coexisting with a radar system using the same frequency band, proposing and assessing uplink reception algorithms,  for both the cases that the radar signal covariance matrix was known at the receiver, 
and that no prior knowledge of the clutter statistics were available. At the analysis stage, we have used information-theoretic arguments to show that, in the large number of BS antennas limit, and with perfect CSI, the system becomes resistant to clutter effects. This fact has been also confirmed by the numerical results, that have shown that, for all the considered data detection structures, the performance steadily increases as the array size at the BS grows larger. 
When instead imperfect CSI is taken into account, the clutter effect does not vanish for large number of antennas, even though the numerical results have confirmed that, for the considered values of $M$, performance improves for increasing $M$ also in this case.
The encouraging results of this paper suggest thus another nice peculiarity of massive MIMO systems, i.e. the resilience to external interference, and should thus attract interest on how to suitably explot massive MIMO to enable the co-existence, in the same frequency band, of both radar and cellular communication systems. This study can indeed be extended along many different tracks. First of all, the paper has considered only the uplink, thus implying that its natural extension regards the consideration of the downlink: it is expected that the BS can in this case play an active role in alleviating the radar interference received by the user devices. Then, this study has not considered, due to lack of space, any resource allocation strategy; it is however expected that improved performance levels might be achieved through proper power allocations algorithms. Finally, a joint co-design, wherein each system tries to achieve its target performance without causing too much harmful interference to the other system, is certainly worth being considered.

\section*{Acknowledgement}
The authors wish to thank Dr. Augusto Aubry (University of Naples "Federico II") for his help with the proof of the theorem reported in Section III, as well as the Editor, Prof. Tolga Duman, for the excellent and fast management of the paper review process.

\appendix

\subsection{SE closed-form expression derivation with PM channel estimation} \label{App_A}
If CM combining with $\mathbf{v}_k^{(q)}=\widehat{\bh}_k^{(q)}$, and PM channel estimation detailed in Section \ref{PM_CE} is assumed, the expectations in Eq. \eqref{SINR_general} can be computed as follows.
\begin{enumerate}
\item \textit{Compute} $\mathbb{E}\left[\widehat{\bh}_k^{(q)\, H} \mathbf{h}_k^{(q)}\right]$: Using Eqs. \eqref{eq:PM_estimation} and \eqref{eq:observable_estimation_q_k} and using the independence between channels, noise and clutter, we can write 
\begin{equation}
\begin{array}{lll}
{E}\left[\widehat{\bh}_k^{(q)\, H} \mathbf{h}_k^{(q)}\right] &= \ds\frac{1}{\sqrt{p_{\rm{p},k}}} \mathbb{E}\left[\mathbf{r}_{q,k}^H {\bh}_k^{(q)} \right] \\ & =\text{tr} \left( \mathbb{E}\left[\mathbf{h}_k^{(q)} {\bh}_k^{(q), \,  H} \right]\right)= M \beta_k^2 \, .
\end{array}
\label{mean_desired_PM}
\end{equation}

\item \textit{Compute} $\mathbb{E}\left[\left|\widehat{\bh}_k^{(q)\, H} \mathbf{h}_j^{(q)}\right|^2\right]$: Using again Eqs. \eqref{eq:PM_estimation} and \eqref{eq:observable_estimation_q_k}, and defining $\mathbf{\bar{c}}_{q,k}=  \mathcal{C}_q \widetilde{\mathbf{P}}_k^{(q)\, *} $ we can write \eqref{mean_interference_kj_PM}, shown on top of the next page, 
\begin{figure*}
\begin{equation}
\begin{array}{llll}
\mathbb{E}\left[\left|\widehat{\bh}_k^{(q)\, H} \mathbf{h}_j^{(q)}\right|^2\right] & = & \ds \frac{1}{p_{\rm{p},k}} \mathbb{E}\left[\left|\left( \ds \sum_{i=1}^K\sqrt{p_{\rm{p},i}}\mathbf{h}_i^{(q)} \mathbf{P}_i^{(q) \, T} \widetilde{\mathbf{P}}_k^{(q)\, *}
+ \widetilde{\mathbf{w}}_{q,k} + \mathbf{\bar{c}}_{q,k} \right) ^H \mathbf{h}_j^{(q)}\right|^2\right] 
 \stackrel{(a)}{=}
\ds \frac{p_{\rm{p},j}}{p_{\rm{p},k}}\mathbb{E}\left[\left| \mathbf{h}_j^{(q)\, H} \mathbf{h}_j^{(q)}\right|^2\right] \left| \mathbf{P}_j^{(q) \, T} \widetilde{\mathbf{P}}_k^{(q)\, *}\right|^2
\\ & &
+ \ds \sum_{\substack{i=1 \\ i \neq j}}^K {\ds \frac{p_{\rm{p},i}}{p_{\rm{p},k}}\mathbb{E}\left[\left| \mathbf{h}_i^{(q)\, H} \mathbf{h}_j^{(q)}\right|^2\right] \left| \mathbf{P}_i^{(q) \, T} \widetilde{\mathbf{P}}_k^{(q)\, *}\right|^2} 
 + \mathbb{E}\left[\left| \widetilde{\mathbf{w}}_{q,k}^H \mathbf{h}_j^{(q)}\right|^2\right] + \mathbb{E}\left[\left| \mathbf{\bar{c}}_{q,k}^H \mathbf{h}_j^{(q)}\right|^2\right] 
 \\  &\stackrel{(b)}{=}&  \ds \frac{p_{\rm{p},j}}{p_{\rm{p},k}}  \beta_j^4 M (M+1) \left| \mathbf{P}_j^{(q) \, T} \widetilde{\mathbf{P}}_k^{(q)\, *}\right|^2 
  + \ds \sum_{\substack{i=1 \\ i \neq j}}^K {\ds \frac{p_{\rm{p},i}}{p_{\rm{p},k}}\beta_j^2 \beta_i^2 M \left| \mathbf{P}_i^{(q) \, T} \widetilde{\mathbf{P}}_k^{(q)\, *}\right|^2 } + \ds \frac{1}{p_{\rm{p},k}} \beta_j^2 \sigma_w^2 M + \ds \frac{1}{p_{\rm{p},k}} \beta_j^2 \text{tr} \left( \widetilde{\mathbf{K}}_{{\rm c},k}^{(q)} \right) 
\\ &
 \stackrel{(c)}{=} & \ds \frac{p_{\rm{p},j}}{p_{\rm{p},k}} \beta_j^4 M^2 \left| \mathbf{P}_j^{(q) \, T} \widetilde{\mathbf{P}}_k^{(q)\, *}\right|^2 + \ds \frac{1}{p_{\rm{p},k}} \beta_j^2  \text{tr} \left( \mathbf{R}_{q,k}\right) \, ,
\end{array}
\label{mean_interference_kj_PM}
\end{equation}
\hrulefill
\end{figure*}
where equality $(a)$ follows from the fact that the variance of a sum of independent RVs is equal to the sum of the variances, equality $(b)$ follows from the independence between channels, noise and clutter and from the relation 
\begin{equation}
\mathbb{E}\left[\left| \mathbf{h}_j^{(q)\, H} \mathbf{h}_j^{(q)}\right|^2\right]= \beta_j^4 M (M+1) ,
\end{equation} 
and, finally, equality  $(c)$ follows from the linearity of the trace operator and from the definitions in Eq. \eqref{R_y}.
\item \textit{Compute} $\mathbb{E}\left[ \norm{\widehat{\bh}_k^{(q)}}^2\right]$: Using again Eqs. \eqref{eq:PM_estimation} and \eqref{eq:observable_estimation_q_k}, we have
\begin{equation}
\mathbb{E}\left[ \norm{\widehat{\bh}_k^{(q)}}^2\right]= \ds \frac{1}{p_{\rm{p},k}} \mathbb{E}\left[\mathbf{r}_{q,k}^H\mathbf{r}_{q,k}\right]=\ds \frac{1}{p_{\rm{p},k}} \text{tr}\left( \mathbf{R}_{q,k} \right)\, .
\label{mean_noise_PM}
\end{equation}
\item \textit{Compute} $\mathbb{E} \left[ \left| \widehat{\bh}_k^{(q)\, H} \mathbf{C}^{(n)} \right|^2\right]$: Using Eq. \eqref{clutter_covariance_ell_n} and the properties of trace operator we obtain
\begin{equation}
\begin{array}{lll}
\mathbb{E} \left[ \left| \widehat{\bh}_k^{(q)\, H} \mathbf{C}^{(n)} \right|^2\right]& =\mathbb{E} \left[  \widehat{\bh}_k^{(q)\, H} \mathbf{C}^{(n)} \mathbf{C}^{(n) \, H} \widehat{\bh}_k^{(q)} \right] \\ & =\ds \frac{1}{p_{\rm{p},k}} \text{tr}\left( \mathbf{R}_{q,k} \mathbf{K}_{\mathbf{C}^{(n)}} \right)
\end{array}
\label{mean_clutter_PM}
\end{equation}
\end{enumerate}
Substituting Eqs. \eqref{mean_desired_PM}, \eqref{mean_interference_kj_PM}, \eqref{mean_noise_PM}, and \eqref{mean_clutter_PM} in Eq. \eqref{SINR_general} we obtain Eq. \eqref{SINR_MR_PM}.

\subsection{SE closed form expression with LMMSE channel estimation} \label{App_B}
If CM combining with $\mathbf{v}_k^{(q)}=\widehat{\bh}_k^{(q)}$, and LMMSE channel estimation detailed in Section \ref{LMMSE_CE} is assumed, the expectations in Eq. \eqref{SINR_general} can be computed as follows.
\begin{enumerate}
\item \textit{Compute} $\mathbb{E}\left[\widehat{\bh}_k^{(q)\, H} \mathbf{h}_k^{(q)}\right]$: Denoting with $\widetilde{\bh}_k^{(q)}=\bh_k^{(q)}-\widehat{\bh}_k^{(q)}$ the channel estimation error, the well known LMMSE estimation property results in the fact that $\widehat{\bh}_k^{(q)}$ and $\widetilde{\bh}_k^{(q)}$ are independent. Using this property and substituting $\bh_k^{(q)}=\widetilde{\bh}_k^{(q)}+\widehat{\bh}_k^{(q)}$ we have
\begin{equation}
\begin{array}{lll}
{E}\left[\widehat{\bh}_k^{(q)\, H} \mathbf{h}_k^{(q)}\right] & =\mathbb{E}\left[\widehat{\bh}_k^{(q)\, H}\left( \widetilde{\bh}_k^{(q)}+\widehat{\bh}_k^{(q)}\right) \right] \\ &=\mathbb{E}\left[\widehat{\bh}_k^{(q)\, H}\widehat{\bh}_k^{(q)}\right]= \text{tr} \left(  \mathbb{E}\left[\widehat{\bh}_k^{(q)}\widehat{\bh}_k^{(q)\, H}\right]\right)\, .
\end{array}
\end{equation}
Using the definitions in Eqs. \eqref{MMSE_estimate} and \eqref{D_qk}, we have:
\begin{equation}
\mathbb{E}\left[\widehat{\bh}_k^{(q)\, H} \mathbf{h}_k^{(q)}\right]=\sqrt{p_{\rm{p},k}} \beta_k^2 \text{tr}\left( \mathbf{D}_{q,k} \right) \, .
\label{mean_desired}
\end{equation}

\item \textit{Compute} $\mathbb{E}\left[\left|\widehat{\bh}_k^{(q)\, H} \mathbf{h}_j^{(q)}\right|^2\right]$: Using Eqs. \eqref{eq:observable_estimation_q_k} and \eqref{MMSE_estimate} we can write \eqref{mean_interference_kj}, shown at the top of next page,
\begin{figure*}
\begin{equation}
\begin{array}{llll}
\mathbb{E}\left[\left|\widehat{\bh}_k^{(q)\, H} \mathbf{h}_j^{(q)}\right|^2\right]&= &\mathbb{E}\left[\left|\left( \ds \sum_{i=1}^K\sqrt{p_{\rm{p},i}}\mathbf{h}_i^{(q)} \mathbf{P}_i^{(q) \, T} \widetilde{\mathbf{P}}_k^{(q)\, *}
+ \widetilde{\mathbf{w}}_{q,k} + \mathbf{\bar{c}}_{q,k} \right) ^H \mathbf{D}_{q,k} \mathbf{h}_j^{(q)}\right|^2\right]  \stackrel{(a)}{=}
p_{\rm{p},j}\mathbb{E}\left[\left| \mathbf{h}_j^{(q)\, H} \mathbf{D}_{q,k} \mathbf{h}_j^{(q)}\right|^2\right] \left| \mathbf{P}_j^{(q) \, T} \widetilde{\mathbf{P}}_k^{(q)\, *}\right|^2
\\ & &
\ds + \sum_{\substack{i=1 \\ i \neq j}}^K {p_{\rm{p},i}\mathbb{E}\left[\left| \mathbf{h}_i^{(q)\, H} \mathbf{D}_{q,k} \mathbf{h}_j^{(q)}\right|^2\right] \left| \mathbf{P}_i^{(q) \, T} \widetilde{\mathbf{P}}_k^{(q)\, *}\right|^2} 
+ \mathbb{E}\left[\left| \widetilde{\mathbf{w}}_{q,k}^H \mathbf{D}_{q,k} \mathbf{h}_j^{(q)}\right|^2\right] + \mathbb{E}\left[\left| \mathbf{\bar{c}}_{q,k}^H\mathbf{D}_{q,k} \mathbf{h}_j^{(q)}\right|^2\right] 
\\ & 
\ds \stackrel{(b)}{=}  &p_{\rm{p},j} \beta_j^4 \left[ \text{tr} \left( \mathbf{D}_{q,k}\right)^2 + \text{tr} \left(\mathbf{D}_{q,k} \mathbf{D}_{q,k}\right)\right] \left| \mathbf{P}_j^{(q) \, T} \widetilde{\mathbf{P}}_k^{(q)\, *}\right|^2 
 + \ds \sum_{\substack{i=1 \\ i \neq j}}^K {p_{\rm{p},i}\beta_j^2 \beta_i^2 \text{tr} \left(\mathbf{D}_{q,k} \mathbf{D}_{q,k}\right) \left| \mathbf{P}_i^{(q) \, T} \widetilde{\mathbf{P}}_k^{(q)\, *}\right|^2 } + \\
 && \beta_j^2 \sigma_w^2 \text{tr} \left(\mathbf{D}_{q,k} \mathbf{D}_{q,k}\right) + \beta_j^2 \text{tr} \left(\mathbf{D}_{q,k} \mathbf{D}_{q,k} \widetilde{\mathbf{K}}_c^{(q)} \right) 
\stackrel{(c)}{=} p_{\rm{p},j} \beta_j^4 \text{tr} \left( \mathbf{D}_{q,k}\right)^2 \left| \mathbf{P}_j^{(q) \, T} \widetilde{\mathbf{P}}_k^{(q)\, *}\right|^2 + \sqrt{p_{\rm{p},k}}\beta_j^2 \beta_k^2 \text{tr} \left( \mathbf{D}_{q,k}\right) \, ,
\end{array}
\label{mean_interference_kj}
\end{equation}
\hrulefill
\end{figure*}
where $(a)$ follows from the fact that the variance of a sum of independent RVs is equal to the sum of the variances, $(b)$ follows from the independence between channels, noise and clutter  and from the relation 
\begin{equation}
\mathbb{E}\!\left[\left| \mathbf{h}_j^{(q) H} \mathbf{D}_{q,k} \mathbf{h}_j^{(q)}
\right|^2\right]\!\!= \!\!\beta_j^4 \!\left[ \text{tr} \left( \mathbf{D}_{q,k}\right)^2\!+ \text{tr} \left(\mathbf{D}_{q,k} \mathbf{D}_{q,k}\right)\right] ,
\end{equation}
 and $(c)$ follows from the linearity of the trace operator and from the definitions in Eqs. \eqref{D_qk} and \eqref{R_y}.
\item \textit{Compute} $\mathbb{E}\left[ \norm{\widehat{\bh}_k^{(q)}}^2\right]$: Using Eq. \eqref{mean_desired} we have
\begin{equation}
\mathbb{E}\left[ \norm{\widehat{\bh}_k^{(q)}}^2\right]=\mathbb{E}\left[\widehat{\bh}_k^{(q)\, H}\widehat{\bh}_k^{(q)}\right]=\sqrt{p_{\rm{p},k}} \beta_k^2 \text{tr}\left( \mathbf{D}_{q,k} \right)\, .
\label{mean_noise}
\end{equation}
\item \textit{Compute} $\mathbb{E} \left[ \left| \widehat{\bh}_k^{(q)\, H} \mathbf{C}^{(n)} \right|^2\right]$: Using Eq. \eqref{clutter_covariance_ell_n} and the properties of trace operator we obtain
\begin{equation}
\begin{array}{llll}
\mathbb{E} \left[ \left| \widehat{\bh}_k^{(q)\, H} \mathbf{C}^{(n)} \right|^2\right]
&=\mathbb{E} \left[  \widehat{\bh}_k^{(q)\, H} \mathbf{C}^{(n)} \mathbf{C}^{(n) \, H} \widehat{\bh}_k^{(q)} \right] \\&=\sqrt{p_{\rm{p},k}} \beta_k^2 \text{tr}\left( \mathbf{D}_{q,k} \mathbf{K}_{\mathbf{C}^{(n)}} \right)
\label{mean_clutter}
\end{array}
\end{equation}
\end{enumerate}
Substituting Eqs. \eqref{mean_desired}, \eqref{mean_interference_kj}, \eqref{mean_noise}, and \eqref{mean_clutter} in Eq. \eqref{SINR_general} we finally obtain Eq. \eqref{SINR_MR_MMSE}.

\bibliographystyle{IEEEtran}
\bibliography{MyReference}


\end{document}